\documentclass[envcountsect,runningheads]{llncs}

\usepackage[T1]{fontenc}
\usepackage[utf8]{inputenc}
\usepackage{graphicx}
\usepackage{xcolor}
\usepackage{hyperref}

\usepackage{todonotes}
\usepackage{nicefrac}
\usepackage{amsmath,amssymb}
\usepackage{mathtools}
\usepackage{cleveref}

\newcommand{\bits}{\{0,1\}}
\newcommand{\BE}{\mathbb{E}}
\newcommand{\BH}{\mathbf H}

\newcommand{\comp}[1]{\overline{#1}}
\newcommand{\conv}{\operatorname{conv}}
\newcommand{\cp}{p}
\newcommand{\CA}{\mathcal A}
\newcommand{\CB}{\mathcal B}
\newcommand{\CD}{\mathcal D}
\newcommand{\CE}{\mathcal E}
\newcommand{\CF}{\mathcal F}
\newcommand{\CG}{\mathcal G}
\newcommand{\CS}{\mathcal S}
\newcommand{\CT}{\mathcal T}
\newcommand{\CV}{\mathcal V}
\newcommand{\CW}{\mathcal W}
\newcommand{\CX}{\mathcal X}
\newcommand{\CY}{\mathcal Y}
\newcommand{\eps}{\varepsilon}
\newcommand{\Ex}[2]{{\mathbb{E}}_{#1}\left[#2\right]}
\newcommand{\FG}{\mathfrak G}
\newcommand{\FW}{\mathfrak W}
\providecommand{\Inf}[2]{\mathbf{I}\left(#1:#2\right)} 
\providecommand{\Infp}[3]{\mathbf{I}_{#3}\left(#1:#2\right)} 
\providecommand{\Infc}[4][]{\mathbf{I}_{#1}\left(#2:#3|#4\right)}
\newcommand{\parity}{\mathrm{parity}}
\newcommand{\Path}{\mathrm{path}}
\newcommand{\ps}[2]{\prescript{}{[#2]}{#1}}
\newcommand{\R}{\mathbb{R}}
\newcommand{\supp}{\mathrm{supp}}
\newcommand{\zeros}{\operatorname{Zeros}}
\newcommand{\Zeros}{\zeros}
\newcommand{\Z}{\mathbb{Z}}

\spnewtheorem{clm}[theorem]{Claim}{\bfseries}{\itshape}

\newcommand{\replabel}{\label} 

\ExplSyntaxOn

\NewDocumentCommand{\repeatlemma}{m}
 {
  \group_begin:
  \renewcommand{\thelemma}{\ref{#1}}
  \renewcommand{\replabel}[1]{\tag{\ref{##1}}}
  \prop_item:Nn \g_replemma_prop { #1 }
  \endlemma
  \group_end:
 }

\NewDocumentEnvironment{replemma}{m+b}
 {
  \prop_gput:Nnn \g_replemma_prop { #1 } { \lemma #2 \endlemma }
  \lemma#2\unskip\label{#1}\endlemma
 }{}

\prop_new:N \g_replemma_prop

\NewDocumentCommand{\repeatclm}{m}
 {
  \group_begin:
  \renewcommand{\theclm}{\ref{#1}}
  \renewcommand{\replabel}[1]{\tag{\ref{##1}}}
  \prop_item:Nn \g_repclm_prop { #1 }
  \endclm
  \group_end:
 }

\NewDocumentEnvironment{repclm}{m+b}
 {
  \prop_gput:Nnn \g_repclm_prop { #1 } { \clm #2 \endclm }
  \clm#2\unskip\label{#1}\endclm
 }{}

\prop_new:N \g_repclm_prop

\ExplSyntaxOff
\begin{document}
\title{Lower Bounds on the Complexity of Mixed-Integer Programs for Stable Set and Knapsack\thanks{Jamico Schade has been funded by the Deutsche Forschungsgemeinschaft (DFG, German Research Foundation) under project 451026932. Stefan Weltge has been supported by DFG under projects 451026932 and 277991500/GRK2201.}}
\titlerunning{MIP formulations for Stable Set and Knapsack}
\author{Jamico Schade\inst{1} \and Makrand Sinha\inst{2}\orcidID{0000-0001-5702-2049} \and Stefan Weltge\inst{1}\orcidID{0000-0002-0102-8326}}
\authorrunning{J. Schade et al.}
\institute{Technical University of Munich, \email{\{jamico.schade,weltge\}@tum.de} \and University of Illinois at Urbana-Champaign, \email{msinha@illinois.edu}}
\maketitle
\begin{abstract}
Standard mixed-integer programming formulations for the stable set problem on $ n $-node graphs require $ n $ integer variables.
We prove that this is almost optimal:
We give a family of $ n $-node graphs for which every polynomial-size MIP formulation requires $ \Omega(n/\log^2 n) $ integer variables.
By a polyhedral reduction we obtain an analogous result for $n$-item knapsack problems.
In both cases, this improves the previously known bounds of $ \Omega(\sqrt{n}/\log n) $ by Cevallos,~Weltge~\&~Zenklusen~(SODA 2018).

To this end, we show that there exists a family of $n$-node graphs whose stable set polytopes satisfy the following: any  $(1+\nicefrac{\varepsilon}{n})$-approximate extended formulation for these polytopes, for some constant $ \varepsilon > 0$, has size $2^{\Omega(n/\log n)}$.
Our proof extends and simplifies the information-theoretic methods due to G\"{o}\"{o}s,~Jain~\&~Watson~(FOCS 2016, SIAM~J.~Comput.~2018) who showed the same result for the case of exact extended formulations (i.e. $\varepsilon = 0$).

\end{abstract}

\section{Introduction}

Combinatorial optimization problems are often expressed by different formulations as mixed-integer programs (MIPs).
A simple example is given by the \emph{matching problem}, which is often described as
\[
    \max \left\{ c^\intercal x : x \in \Z^E_{\ge 0}, \, \sum \nolimits_{e \in \delta(v)} x(e) \le 1 \ \text{ for all } v \in V \right\},
\]
where $G = (V,E)$ is the complete undirected graph on $n$ nodes and $\delta(v)$ is the set of edges in $G$ that are incident to $v$.
This formulation is attractive in the sense that it only consists of a small number of linear constraints, which naturally reflect the definition of a matching.
However, it comes at the cost of $|E| = \Theta(n^2)$ integer variables.
Since the performance of algorithms for solving integer programs is much more sensitive in the number of integer variables than in the number of constraints, the question arises whether there are MIP formulations for the matching problem that use significantly fewer integer variables, yet a reasonable (say, polynomial in $n$) number of constraints.
Note that a formulation without any integer variables can be obtained by adding linear inequalities that completely describe the matching polytope of $G$, but this would require exponentially many constraints~\cite{edmonds1965maximum}.
However, there is a simple (lesser-known) linear-size MIP formulation for the matching problem that only uses $n$ integer variables:
If $D = (V,A)$ is a digraph that arises from $G$ by orienting its edges arbitrarily, a valid MIP formulation for the matching problem is
\[
    \max \left\{ c^\intercal x :
    x \in \R^A_{\ge 0}, \,
    y \in \Z^V, \,
    x(\delta(v)) \le 1 \text{ and }
    x(\delta^{\mathrm{in}}(v)) = y(v) \text{ for all } v \in V \right\},
\]
where $\delta^{\mathrm{in}}(v)$ denotes the set of arcs of $D$ that enter $v$, and $x(F) = \sum_{e \in F} x(e)$.
For the correctness of this formulation, see~\cite[Prop.~6.2]{Cevallos}.

As another example, consider the symmetric \emph{traveling salesman problem} over $G$.
Standard MIP formulations for this problem contain at least one integer variable for each edge, again resulting in $\Theta(n^2)$ integer variables.
However, it is possible to come up with MIP formulations for the traveling salesman problem that use only $O(n \log n)$ integer variables and still have polynomial size, see~\cite[Cor.~50]{arxiv}.
Other combinatorial optimization problems such as the \emph{spanning tree problem} even admit polynomial-size MIP formulations without any integer variables, so-called extended formulations (see, e.g., \cite{conforti2013extended,Kaibel}).

The above examples illustrate that it is usually not obvious how many integer variables are needed in small-size MIP formulations of combinatorial optimization problems.
Moreover, they refer to problems for which there exist polynomial-size MIP formulations, which use considerably fewer integer variables than the standard formulations.
In this work, we consider two prominent combinatorial optimization problems, for which such formulations are not known.
The first is the \emph{stable set problem} over a general undirected $n$-node graph $G = (V,E)$, which is usually described as
\begin{equation}
    \label{eqMIPStableSet}
    \max \left\{ c^\intercal x : x \in \{0,1\}^V, \, x(v) + x(w) \le 1 \text{ for all } \{v,w\} \in E \right\},
\end{equation}
and the second is the \emph{knapsack problem}, typically given by
\begin{equation}
    \label{eqMIPKnapsack}
    \max \left\{ c^\intercal x : x \in \{0,1\}^n, \sum \nolimits_{i=1}^n a(i) x(i) \le \beta \right\},
\end{equation}
where $a \in \R^n$ are given item sizes and $\beta \in \R$ is the given capacity.
Both (standard) formulations have $n$ integer variables, and no polynomial-size MIP formulations with $o(n)$ integer variables are known.
Our main motivation for considering these two problems is that Cevallos, Weltge \& Zenklusen~\cite{Cevallos} proved that the number of integer variables in the aforementioned MIP formulations for the matching problem and the traveling salesman problem is optimal up to logarithmic terms, while an almost quadratic gap remained for the case of the stable set problem and the knapsack problem.

To address this claim formally, let us specify what we mean by a MIP formulation for a combinatorial optimization problem.
Here, we consider a combinatorial optimization problem as a pair $(\CV,\CF)$ where $\CV$ is a finite ground set and $\CF$ is a family of (feasible) subsets of $\CV$.
Given weights of $w : \CV \to \R$, the goal is to find a set $S \in \CF$ maximizing $w(S) = \sum_{v \in S} w(v)$.
Now, a MIP formulation for $(\CV,\CF)$ is defined as follows.
First, its feasible region $\Gamma$ should only depend on $(\CV,\CF)$ (and not on the weights to be maximized).
Second, $\Gamma$ should be described by linear inequalities and equations, and a subset of variables that is constrained to integer values.
To this end, we represent $\Gamma$ by a polyhedron $Q \subseteq \R^d$ and an affine map $\sigma : \R^d \to \R^k$ by setting
$
    \Gamma = \Gamma(Q,\sigma) = \{ x \in Q : \sigma(x) \in \Z^k \}
$.
Third, we want to identify each feasible subset $S$ with a point $x_S \in \Gamma$.
However, we do not require the variables to be directly associated with the elements of the ground set, and in particular allow $d \ne |\CV|$.
Finally, node weights should translate to (affine) linear objectives in a consistent way:
We require that for each $w : \CV \to \R$ there is an affine map $c_w : \R^d \to \R$ such that the weight of every feasible set $S$ satisfies $w(S) = c_w(x_S)$.
Note that $\Gamma$ does not necessarily only contain the points $x_S$.
However, when maximizing $c_w$ over $\Gamma$, we require that the optimum is still attained in a point $x_S$, i.e., $\max\{c_w(x) : x \in \Gamma\} = \max\{c_w(x_S) : S \in \CF\} = \max\{w(S) : S \in \CF\}$.
Notice that $\max\{c_w(x) : x \in \Gamma\}$ can be formulated as a MIP.

If $(Q,\sigma)$ satisfies the above properties, we say that it is a MIP formulation for $(\CV,\CF)$.
We define the size of $(Q,\sigma)$ to be the number of facets of $Q$ (number of linear inequalities needed to describe $Q$) and say that it has $k$ integer variables.

With this notion, it is shown\footnote{Actually, our definition of a MIP formulation slightly differs from the notion in \cite{Cevallos}. However, both definitions are equivalent, see \Cref{lemMIPdef}.} in~\cite{Cevallos} that every subexponential-size MIP formulation for the matching problem or traveling salesman problem has $\Omega(\nicefrac{n}{\log n})$ integer variables.
Moreover, they proved that there exist $n$-node graphs and $n$-item instances for which any subexponential-size MIP formulation for the stable set problem or the knapsack problem, respectively, has $\Omega(\nicefrac{\sqrt{n}}{\log n})$ integer variables.
In this work, we close this gap and show that, for these two problems, the standard MIP formulations \eqref{eqMIPStableSet} and \eqref{eqMIPKnapsack} already use an (up to logarithmic terms) optimal number of integer variables:

\begin{theorem}
    \label{thmMain}
    There is a constant $c > 0$ and a family of graphs (knapsack instances) such that every MIP formulation for the stable set (knapsack) problem of an $n$-node graph ($n$-item instance) in this family requires $\Omega(\nicefrac{n}{\log^2 n})$ integer variables, unless its size is at least $2^{cn/\log n}$.
\end{theorem}

We will mostly focus on proving the above result for the stable set problem.
In fact, using a known polyhedral reduction, we will show that it directly implies the result for the knapsack problem.
In order to obtain lower bounds on the number of integer variables in MIP formulations, the authors in~\cite{Cevallos} observed that every MIP formulation can be turned into an approximate extended formulation of the polytope $P(\CV,\CF) = \conv \{ \chi(S) : S \in \CF \}$, where $\chi(S) \in \{0,1\}^{\CV}$ is the characteristic vector of $S$.
In terms of the stable set problem over a graph $G$, $P(\CV,\CF)$ is the stable set polytope of $G$.
Here, an \emph{$\alpha$-approximate extended formulation of size $m$} of a polytope $P$ is (a linear description of) a polyhedron with $m$ facets that can be linearly projected onto a polytope $P'$ with $P \subseteq P' \subseteq \alpha P$.
For instance, it is shown in~\cite{Cevallos} that, for every $\eps > 0$, every MIP formulation for the stable set problem of size $m$ with $k$ integer variables can be turned into a $(1+\eps)$-approximate extended formulation of size $m \cdot (1+\nicefrac{k}{\eps})^{O(k)}$ of the stable set polytope of the corresponding graph (see \cite[Thm.~1.3 \& Lem.~2.2]{Cevallos}).
In light of this result, \Cref{thmMain} follows from the following result about approximate extended formulations which is the main contribution of this work:
\begin{theorem}
    \label{thmEfStableSet}
    There is some constant $\eps > 0 $ and a family of $n$-node graphs $H$ such that any $(1 + \nicefrac{\eps}{n})$-approximate extended formulation of the stable set polytope of $H$ has size $2^{\Omega(n/\log n)}$.
\end{theorem}
Note that the statement above considers approximations with very small error. Lower bounds for approximate formulations with small error are known for several prominent polytopes that arise in combinatorial optimization such as the matching polytope (see Braun \& Pokutta~\cite{braun2014matching}, Rothvoß~\cite[Thm. 16]{rothvoss2017matching}, and Sinha~\cite{sinha2018lower}) or the cut polytope (see Braun, Fiorini, Pokutta \& Steurer~\cite{braun2015approximation} and \cite[\S 5.2]{arxiv}). For the stable set problem on the other hand, existing lower bounds apply to much larger (say, constant) errors but do not yield bounds that are close to exponential (e.g. \cite{bazzi2019no}). For this reason, the aforementioned lower bound on MIP formulations in \cite{Cevallos} did not directly rely on approximate extended formulations for stable set polytopes but instead was obtained by transferring results about the cut polytope (see Fiorini, Massar, Pokutta, Tiwary \& de Wolf~\cite[Lem. 8]{fiorini2015exponential}).

Unfortunately, reductions to the cut or matching polytopes can only show that the sizes of approximate extended formulations must be exponential in $\sqrt{d}$, where $d$ is the dimension of the corresponding polytope. For the case of the matching problem and the traveling salesman problem, we have $\sqrt{d} = O(n)$ and hence these bounds are sufficient to prove optimality of the aforementioned MIP formulations. However, in the case of the stable set problem, we have $\sqrt{d} = \sqrt{n}$, which is the reason why the authors of \cite{Cevallos} were only able to prove that subexponential-size MIP formulations for the stable set problem have $\Omega(\nicefrac{\sqrt{n}}{\log n})$ integer variables. We are able to circumvent this $\sqrt{d}$ bottleneck since we directly work with the stable set polytopes given by \Cref{thmEfStableSet}: here any extended formulation must be of size exponential in the dimension $d$ (ignoring logarithmic factors). 

We prove \Cref{thmEfStableSet} by extending and simplifying the techniques of Göös, Jain \& Watson~\cite{Goos} who proved \Cref{thmEfStableSet} for the case of \emph{exact} ($\eps = 0$) extended formulations.
They were inspired by connections to Karchmer-Wigderson games made in \cite{Hrubes}. By relying on reductions from certain constraint satisfaction problems (CSPs) that arise from Tseitin tautologies, they constructed a family of stable set polytopes that they showed have $2^{\Omega(n/\log n)}$ extension complexity.
Their proof uses information complexity arguments building on the work of Huynh and Nordstr\"om~\cite{HN12} and is significantly involved.
Furthermore, the proof departs from other approximate extended formulations lower bounds for the cut and matching polytopes, which can be obtained by a fairly unified framework (see \cite[Ch. 12]{RY20}).

Our proof is still based on the main ideas of \cite{Goos} but extending it to the approximate case is quite involved. To prove our lower bound, we follow the common information framework introduced by \cite{BM13} and further developed in \cite{BP16,BP16a,sinha2018lower}, mostly
in the interest of bridging the gap to previous proofs for other approximate extended formulation lower bounds.
Along the way, we simplify several parts of the information-theoretic arguments in \cite{Goos} and also show that the family of stable set polytopes used in \Cref{thmEfStableSet} are in fact simple to describe explicitly (without going through CSP reductions).
Separately from this, as mentioned before, we show that this also implies a lower bound for knapsack MIP formulations via a standard polyhedral reduction.

We remark that since information complexity arguments are typically robust to approximations, we believe that with some amount of work, the approach in \cite{Goos} can be made to yield \Cref{thmEfStableSet} above.
However, our proof is very much in line with the previous lower bounds for the cut and matching polytope which ultimately  reduce the problem to understanding the \emph{nonnegative rank} of a certain matrix called the \emph{unique disjointness} matrix, possibly through a randomized reduction. Our proof in fact suggests that one may be able to prove all of these lower bounds --- for matching, stable set, and knapsack polytopes --- in a unified way, via a randomized reduction to unique disjointness. We leave this as an interesting open question for follow-up work.

\paragraph{Structure.}
We start by describing a family of graphs $H$ for which \Cref{thmEfStableSet} holds in \Cref{secinstances}.
Moreover, we derive a matrix $S$ whose nonnegative rank will give a lower bound on the size of any $(1+\frac{\eps}{n})$-approximate extended formulation of the stable set polytope of $H$.
Our proof strategy for obtaining a lower bound on the nonnegative rank of $S$ is explained in \Cref{seclowbound}.
Notions and statements on information theory as well as technical details of the main proof are presented in \cref{secPrelim} and \cref{secProofKeyThm}.
\section{Instances}
\label{secinstances}
\subsection{Graph family}

The graphs $H$ in the statement of \Cref{thmEfStableSet} will arise from a family of sparse graphs $G = (V,E)$ with an \emph{odd} number of nodes and a certain connectivity property to be defined later.

To define $H = H(G)$, let us fix a set of colors $\CX = \bits^3$.
Each node $v \in V$ is lifted to several copies, so that each copy corresponds to a different coloring of the edges incident to $v$, i.e., the nodes of $H$ are the pairs $(v,x^v)$ where $v \in V$ and $x^v$ denotes a coloring of the incident edges of $v$ with colors in $\CX$.

In $H$, all copies of a single node $v$ form a clique $C_v$.
Moreover, if $v,w \in V$ are connected by an edge $e \in E$, then we also draw an edge between copies of $v$ and $w$ in $H$ if they label $e$ with a different color, see \Cref{figGraphOutline}.

\begin{figure}
    \centering
    \includegraphics[width=\textwidth]{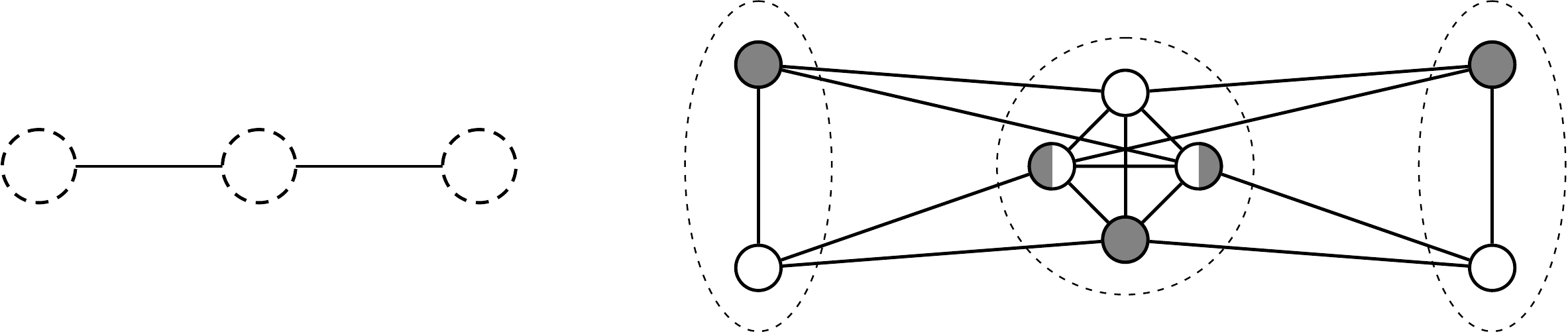}
    \caption{Construction of the graph $H$ (right), where $G$ is a path on three nodes (left) and $\CX$ consists of only two colors (white and gray).}
    \label{figGraphOutline}
\end{figure}

Note that stable sets in $H$ can be obtained in a simple way:
Pick any (global) coloring $x \in \CX^E$ of the edges of $G$.
Now, for each node $v \in V$, select the (unique) node in $C_v$ that colors the edges incident to $v$ according to $x$.
This yields a maximal stable set in $H$.
In fact, every maximal stable set in $H$ arises in this way.

Moreover, observe that if $G$ has constant degree, then the number of nodes and the number of edges of $H$ are both linear in $|V|$ (since $\CX$ has constant size).

\subsection{Nonnegative rank and partial slack matrices}

Let $P$ denote the stable set polytope of the graph $H$ and consider any polytope $P'$ with $P \subseteq P' \subseteq (1+\frac{\eps}{n}) P$.
A common approach to obtain a lower bound on the number of facets of $P'$ is based on considering partial slack matrices of $P'$.
A \emph{partial slack matrix} of $P'$ is a (nonnegative) matrix $S$, where each row $i$ corresponds to some point $x^i \in P'$ and each column $j$ corresponds to a linear inequality $a_j^\intercal x \le b_j$ that is satisfied by all $x \in P'$, such that each entry $S_{ij}$ is equal to the \emph{slack} of $x^i$ with respect to $a_j^\intercal x \le b_j$, i.e., $S_{ij} = b_j - a_j^\intercal x^i$.
From the seminal paper~\cite{yannakakis1988expressing} of Yannakakis it follows that the size of any extended formulation of $P$ is at least the \emph{nonnegative rank} of $S$, which is the smallest number $r_+$ such that $S$ can be written as the sum of $r_+$ nonnegative rank-$1$ matrices.

\subsection{Gadget and a particular matrix}

Setting $\CY = \bits^3$, we will consider a particular matrix $S = S(G, \eps)$ whose rows and columns are indexed by vectors $x \in \CX^E$ and $y \in \CY^E$, respectively.
The entries of $S$ are not only based on $G,\eps$ but also on a \emph{gadget function} $g: \CX \times \CY \to \bits$ defined via
\begin{equation}
    \label{eqGadget}
    g((x_1,x_2,x_3),(y_1,y_2,y_3)) = x_1 + y_1 + x_2y_2 + x_3y_3 \pmod 2.
\end{equation}
Given a $x \in \CX^E,y \in \CY^E$, one can apply this gadget to obtain a bit string in $\bits^E$ that labels each edge in $E$.
Summing up the parity of each edge incident on $v$ induces a \emph{parity on a node} defined as
\[ \parity(v) = \sum \nolimits_{e \in \delta(v)} g(x_e,y_e) \mod 2.\]
We set 
\[
    \zeros(x,y) = \left\{v \in V : \parity(v) = 0 \right\}
\]
to be the set of nodes with parity zero.
With this, the entries of $S \in \R^{\CX^E \times \CY^E}$ are given by
\begin{equation}
    \label{eqWednesday}
    S_{xy} = |\zeros(x,y)| - 1 + \eps.
\end{equation}
Note that $\sum_{v \in V} \parity(v)$ is always even since the bit label of each edge is summed twice.
This implies that $|V \setminus \zeros(x,y)|$ is even for each $x \in \CX^E,y \in \CY^E$, and since the number of nodes $|V|$ is odd, it follows that $|\zeros(x,y)|$ is always odd and hence $|\zeros(x,y)| \ge 1$.
In particular, we see that every entry of $S$ is positive.

\begin{lemma}
    \label{lemK8sgduw}
    $S$ is a partial slack matrix of every polytope $P'$ with $P \subseteq P' \subseteq (1+\nicefrac{\eps}{|V(H)|}) P$, where $P$ is the stable set polytope of $H$.
\end{lemma}
\begin{proof}
To show that $S$ is a partial slack matrix of $P'$, we have to define a point $z^x \in P'$ for every $x \in \CX^E$ and associate a linear inequality to every $y \in \CY^E$ that is valid for $P'$ and such that the slack of $z^x$ with respect to the inequality associated to $y$ is equal to $S_{xy}$.
For $x \in \CX^E$ we define the set
\[
    S_x = \{ (v,x|_{\delta(v)}) : v \in V \}
\]
and denote by $z^x \in \{0,1\}^{V(H)}$ the characteristic vector of $S_x$.
Note that $S_x$ is a stable set in $H$ and hence $z^x \in P \subseteq P'$.
For $y \in \CY^E$ we define the set
\[
    U_y = \left\{ (v,x) \in V(H) : \sum \nolimits_{e \in \delta(v)} g(x_e,y_e) \text{ is odd}\right\}
\]
and consider the linear inequality
\begin{equation}
    \label{eq43z3g8} \sum \nolimits_{u \in U_y} z_u \le |V| - 1 + \eps.
\end{equation}
Note that the slack of $z^x$ with respect to this inequality is equal to
\begin{align*}
    |V| - 1 + \eps - \sum \nolimits_{u \in U_y} z_u
    & = |V| - 1 + \eps - |S_x \cap U_y| \\
    & = |V| - 1 + \eps - (|V| - |\zeros(x,y)|) = S_{xy}.
\end{align*}
Thus it remains to show that \eqref{eq43z3g8} is satisfied by every $z \in P'$.
To see this, we first claim that every stable set $S$ in $H$ satisfies $|S \cap U_y| \le |V| - 1$.
Indeed, if $ |S| \le |V| - 1 $, then the claim is trivial.
Otherwise, it is easy to see that $S = S_x$ for some $x \in \CX^E$ and hence $|S_x \cap U_y| = |V| - |\zeros(x,y)| \leq |V| - 1$.
This means that the linear inequality $ \sum \nolimits_{u \in U_y} z_u \le |V| - 1 $ is satisfied by every point $z \in P$.
Since $P' \subseteq (1+\nicefrac{\eps}{|V(H)|}) P$, every point $z \in P'$ must satisfy
\[ \sum \nolimits_{u \in U_y} z_u \le \left(1 + \tfrac{\eps}{|V(H)|}\right)(|V| - 1) = |V| - 1 + \frac{\eps}{|V(H)|} \cdot (|V| - 1) \leq |V| - 1 + \eps, \]
which yields~\eqref{eq43z3g8}. \qed
\end{proof}

\subsection{Choice of $G$}

We will see that, for particular choices of $G$, the nonnegative rank of $S$ is large.
To this end, we say that a graph $G$ is called \emph{$(2k+3)$-routable} if there exists a subset of $2k+3$ nodes, called \emph{terminals}, such that for every partition of the terminals into a single terminal and $k+1$ pairs $(v_0,w_0),\dots,(v_k,w_k)$ of terminals there exist edge-disjoint paths $P_0,\dots,P_k$ such that $P_i$ connects $v_i$ and $w_i$.
Note that each $P_i$ depends on the entire partition and not only on the pair $(v_i,w_i)$.

Our main result is the following.
\begin{theorem}
    \label{thmNonnegRank}
    There is some $\eps > 0$ such that if $G$ is $(2k+3)$-routable, then the nonnegative rank of $S = S(G,\eps)$ is $2^{\Omega(k)}$.
\end{theorem}

To obtain \Cref{thmEfStableSet} we choose $G = (V,E)$ to be any \emph{constant-degree} $(2k+3)$-routable graph with an odd number of nodes, and $k = \Theta(|V| / \log |V|)$. An infinite family of such graphs can be obtained by taking any sufficiently strong constant-degree expander (e.g. a Ramanujan graph) (\cite{frieze2001edge,frieze2000optimal}, see also \cite[\S15]{nordstrom2009new}). Recall that $H$ has $\Theta(|V|)$ nodes (since $G$ has constant degree).
\section{Proof strategy}
\label{seclowbound}
\subsection{Nonnegative rank and mutual information}

In order to prove \Cref{thmNonnegRank}, we follow an information-theoretic approach 
introduced by Braverman and Moitra~\cite{BM13} that has already been used in previous works on extended formulations~(c.f.~\cite{braun2014matching,BP16,sinha2018lower}).
To this end, suppose that the nonnegative rank of $S$ is $r_+$, in which case we can write $S = \sum_{r=1}^{r_+} R^{(r)}$ for some nonnegative rank-$1$ matrices $R^{(1)},\dots,R^{(r_+)}$.
Consider the discrete probability space with random variables $X,Y,R$ and distribution $q(X=x, Y=y, R=r) = \nicefrac{R^{(r)}_{xy}}{\|S\|_1}$, where $x$ and $y$ range over the rows and columns of $S$, respectively, and $r \in [r_+]$.
Notice that the marginal distribution of $X,Y$ is given by the (normalized) matrix $S$, i.e., $q(X=x, Y=y) = \nicefrac{S_{xy}}{\|S\|_1}$. Moreover, since each $R^{(r)}$ is a rank-$1$ matrix, we see that $X,Y$ are independent when conditioned on $R$. Thus, we obtain the following proposition.

\begin{proposition}\label{prop:nrank}
    If the nonnegative rank of $S$ is $r_+$, then there is a random variable $R$ with $|\supp(R)|=r_+$, such that $X$ and $Y$ are independent given $R$ and the marginal distribution of $X$ and $Y$ is given by the normalized matrix $S$.
\end{proposition}

Note that the above also implies that $R$ breaks the dependencies between $X$ and $Y$ even if we further condition on any event in the probability space that is a \emph{rectangle}, i.e. any event where $(X,Y) \in \CS \times \CT$ where $\CS \subseteq \CX^E$ and $\CT \subseteq \CY^E$ are subsets of rows and columns respectively.

To prove a lower bound on the nonnegative rank, we shall show that if the support of $R$ was too small, then the probability of certain events in the probability space would be quite different from that given by the distribution $q(X,Y)$ above. In order to do this, instead of directly bounding the size of the support of $R$, it will be more convenient for us to work with certain information-theoretic quantities that give a lower bound on the (logarithm of) the size of the support of $R$:

For a random variable $A$ over a probability space with distribution $p$, we denote the binary \emph{entropy} of $A$ (with respect to $p$) by $\BH_p(A)$.
For any two random variables $A$ and $B$, the entropy of $A$ conditioned on $B$ is defined as $\BH_p(A|B) = \BE_{\cp(b)}[\BH_p(A|b)]$, where $\Ex{p(b)}{f(b)}$ denotes the expected value of a function $f(b)$ under the distribution $p(B)$.
Given an event $\CW$, we write $\BH_p(A|B,\CW) = \BH_q(A|B)$ where $q = p(\cdot \mid \CW)$ arises from $p$ by conditioning on the event $\CW$.
The \emph{mutual information} between $A,B$ is defined as
$\Infp{A}{B}{p} = \BH_p(A) - \BH_p(A|B)$.
Further, the \emph{conditional mutual information} is defined as $\Infc[p]{A}{B}{C} = \BH_p(A|C) - \BH_p(A|BC)$.
Finally, given an event $\CW$, we define $\Infc[p]{A}{B}{C,\CW} = \Infc[q]{A}{B}{C}$, where $q = p(\cdot \mid \CW)$.
We will elaborate on the notation as well as the essential properties of the above quantities in \cref{secPrelim}.

Given $q,X,Y,R$ as in \Cref{prop:nrank}, in what follows we will introduce some further specific random variables $T,W$, and an event $\CD$ that satisfy the following.

\begin{theorem}
    \label{thm99dhsuao}
    There is some $\eps > 0$ such that $\Infc[q]{R}{XY}{TW,\CD} = \Omega(k)$.
\end{theorem}

A basic property of (conditional) mutual information~\cite[Thm.~2.6.4]{CT06}
states that if $\Infc[q]{R}{XY}{TW,\CD} \ge \ell$, then $|\supp(R)| \ge 2^\ell$, and hence this directly yields \Cref{thmNonnegRank}.
In the remainder of this section, we will define $T,W,\CD$ and describe the proof outline of \Cref{thm99dhsuao}.

\subsection{Random pairings and windows}

The random variable $T$ will denote a uniformly random partition of the terminals of $G$ into a single terminal $T_u$ and an ordered list of $k+1$ unordered pairs $T_0, \cdots, T_k$.

Following \cite{Goos}, we say that a $2\times 1$ (horizontal) or $1\times 2$ (vertical) rectangle $w$ of $\CX \times \CY$ is a \emph{$b$-window} if the value of $g$ on the two inputs in $w$ is equal to $b$.
The random variable $W = (W_e)_{e \in E}$ will denote a uniformly random window for every edge of $G$.

An important property of the gadget $g$ is the following.
Given a (horizontal) $b$-window $w = \{(x,y), \, (x, y')\}$, there exist unique $\tilde x \in \CX$, $\tilde y, \tilde y' \in \CY$ such that 
$\ps{w}{00} := \{(\tilde{x},\tilde{y}), \, (\tilde{x},\tilde{y}')\}$,
$\ps{w}{10} := \{(x,\tilde{y}), \, (x,\tilde{y}')\}$, and
$\ps{w}{01} := \{(\tilde{x},y), \, (\tilde{x},y')\}$ are $\comp{b}$-windows.
Thus, we may view $\ps{w}{00}$, $\ps{w}{10}$, $\ps{w}{01}$, $\ps{w}{11} := w$ as a ``stretched'' \textsf{AND} if $w$ is a $1$-window, or a ``stretched'' \textsf{NAND} if $w$ is a $0$-window, see Figure~\ref{figWindows} for an illustration.
Since $g$ is symmetric with respect to its inputs, we may define $\ps{w}{00}$, $\ps{w}{10}$, $\ps{w}{01}$, and $\ps{w}{11}$ analogously for vertical windows.

\begin{figure}
    \centering
    \includegraphics[width=0.8\textwidth]{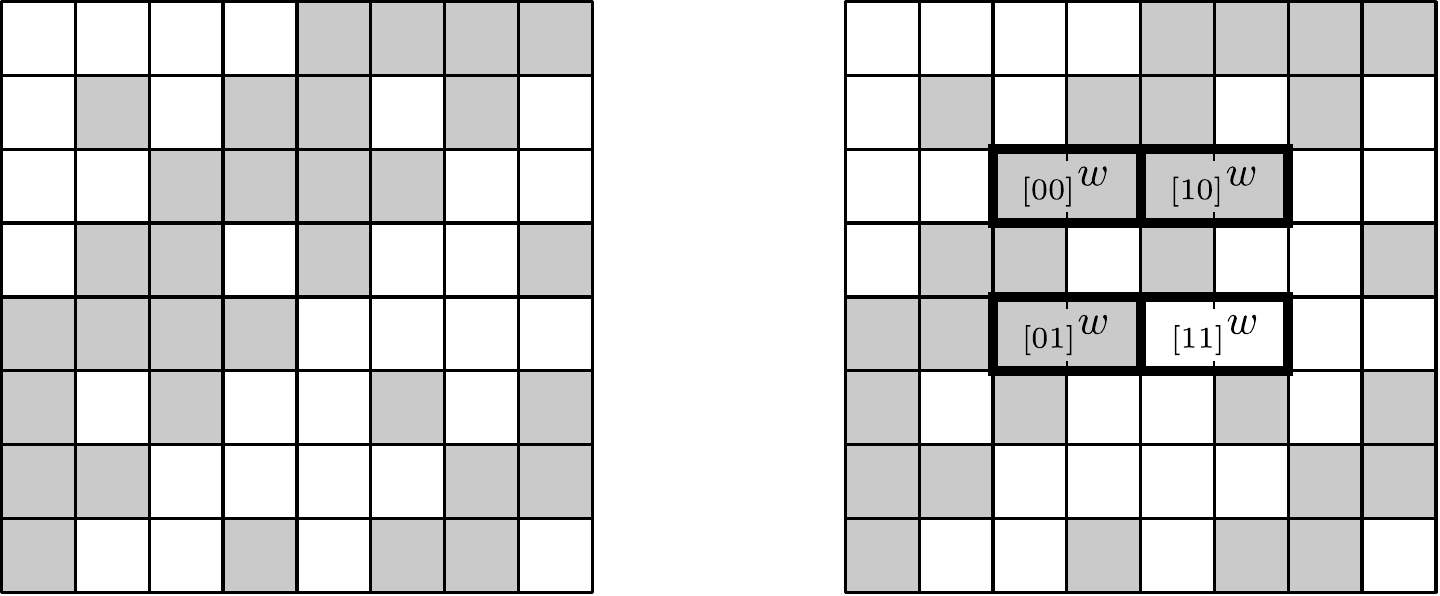}
    \caption{Both pictures show the values of the gadget $g$ in white ($0$) and gray ($1$), where the columns and rows corresponds to the inputs in the order $(0,0,0), \, (0,0,1), \, (0,1,0)$, etc. The right picture shows a $0$-window $w = \ps{w}{11}$ and the corresponding $1$-windows $\ps{w}{00}$, $\ps{w}{01}$, and $\ps{w}{10}$.}
    \label{figWindows}
\end{figure}

In what follows, we will consider the event 
\[
    \CD : \zeros(X, Y) = T_u \cup T_0 \text{ and } (X, Y) \in W,
\]
which, in particular, enforces that the parity of the nodes is zero only on the three terminals in $T_u \cup T_0$.


\subsection{Main argument}

First, recall that for every partition $t$ of the $(2k + 3)$ terminals of $G$ into a single terminal $t_u$ and pairs of terminals $t_0, \dots, t_k$, there exist edge-disjoint paths $P_0,\dots,P_k$ such that $P_i$ connects the terminals $t_i$, which we denote by $\Path_i(t) = P_i$.
We denote by $X^i$ the restriction of $X$ to the indices that correspond to edges in $\Path_i(T)$, i.e., $X^i = (X_e)_{e \in \Path_i(T)}$.
Similarly, we define $X^{-i} = (X_e)_{e \notin \Path_i(T)}$. We use the analogous notation for other random variables as well, e.g. $Y^i, W^i$.

The event $\CD$ is quite useful because it turns out that within this event, for any fixed partition $T=t$ and window $W=w$, the random variables $X^1Y^1, X^2Y^2, \ldots, X^k Y^k$ are mutually independent.
This allows one to use the powerful superadditivity property of mutual information (see \Cref{secSuperadditivity}):
\begin{equation}
	\label{equhsd8g3g}
    \Infc[q]{R}{XY}{TW,\CD} \ge \sum_{i=1}^{k} \Infc[q]{R}{X^iY^i}{TW,\CD}.
\end{equation}
In view of \Cref{thm99dhsuao}, it suffices to focus on the marginal distribution of $X^iY^i$ and show that each of the mutual information terms on the right is $\Omega(1)$.
This can be generated from a much smaller rectangle (submatrix) of the slack matrix $S$: the event $\CD$ fixes a window $W$ such that for all $X,Y \in W$, the parity of all the terminals in $T_1,\cdots, T_k$ is one and $\zeros(X,Y) = T_u \cup T_0$. One could consider the event $\CD$ as part of a larger rectangle where one is allowed to have parity zero on any of the other pair of terminals $T_1, \cdots, T_k$ as well. In contrast, in the marginal probability space for $i \in [k]$, it will suffice to consider the event $\CD$ as part of a smaller rectangle where $\zeros(X,Y) = T_u \cup T_0 \cup T_i$ or $\zeros(X,Y) = T_u \cup T_0$.
In order to do this, note that if we flip the edge labels along any path in the graph, i.e. replace the label $b_e:= g(x_e,y_e)$ for each edge on the path with $\overline{b_e}$, then the parity of the end points of the path flips, while the parity of all other nodes remains unchanged. Since we are conditioning on windows as well, we will switch to a different window in order to flip the edge labels.

For this we rely on the fact that $w \mapsto \ps{w}{00}$ (as defined previously) is a bijection, which allows us to obtain various equivalent but correlated ways of generating the same event.
In particular, consider the following event for each $i \in [k]$, 
\[
    \CD_i: \zeros(X, Y) = T_u \cup T_0 \text{ and } (X, Y) \in \ps{W^i}{00}W^{-i},
\]
where $\ps{W^i}{ab}W^{-i}$ arises from $W$ by replacing each entry $w$ that corresponds to an edge in $\Path_i(T)$ by $\ps{w}{ab}$.

Note that if we have $\zeros(\ps{W^i}{00}W^{-i}) = T_u \cup T_0$, then $\zeros(\ps{W^i}{01}W^{-i}) = \zeros(\ps{W^i}{10}W^{-i}) = T_u \cup T_0$ as well as $\zeros(W) = T_u \cup T_0 \cup T_i$.
Although the event $\CD_i$ can be considered a part of this smaller rectangle, note that $\CD_i$ is a not a subset of $\CD$.
However, since $W \mapsto \ps{W^i}{00}W^{-i}$ is a bijection over the space of windows, the mutual information quantities remain the same.
Thus, by \eqref{equhsd8g3g} we obtain $\Infc[q]{R}{XY}{TW,\CD} \ge \sum_{i=1}^{k} \Infc[q]{R}{X^iY^i}{TW,\CD_i}$ (see again \Cref{secSuperadditivity}), and hence it suffices to prove:

\begin{theorem}
    \label{thmLowerBoundDi}
    For $\eps > 0$ small enough, there is a constant $c > 0$ such that $ \Infc[q]{R}{X^iY^i}{TW,\CD_i} \geq c$ holds for every $i \in [k]$.
\end{theorem}

We elaborate on the main steps of proving the above theorem.
A detailed proof is given in \cref{secProofKeyThm}.
The argument will be the same for each $i \in [k]$, so let us fix an $i \in [k]$.
To obtain a bound on $\Infc[q]{R}{X^iY^i}{TW,\CD_i}$, we will consider a smaller probability space $p$.
To this end, we will treat the windows corresponding to the edges of path between the terminals in $T_i$, denoted $W^i$, differently from the windows for the rest of the edges.

Recall that any window can be seen as a ``00'', ``01'', ``10'' or ``11'' input of (a unique ``stretched AND'' or ``stretched NAND'' of) the gadget $g$.
In particular, we may view $W^i = \ps{W^i}{11}$ as being embedded as a ``11'' input, which also yields three further correlated disjoint random windows that we call $\ps{W^i}{00}$, $\ps{W^i}{01}$, and $\ps{W^i}{10}$.

The probability space $p$ arises as follows.
First, we define random variables $A,B \in \bits$ and restrict to the event that $X^iY^i \in \ps{W^i}{AB}$ and $X_eY_e \in W_e$ for all remaining edges $e$.
Moreover, we fix parts of $T$ and the set of nodes which have parity zero under the labeling given by $W$.
We will see that for any $x,y,a,b$ in the above probability space, we have $|\Zeros(x,y)| = 3$ if $a=0$ or $b=0$ and $|\Zeros(x,y)| = 5$ otherwise. The probability space $p$ will be obtained by restricting to this submatrix and normalizing the entries. In view of \eqref{eqWednesday}, this yields
\begin{equation}\label{eqn83gdhdls}
    p(A=0,B=0) = \frac{2+\eps}{10+4\eps} \approx \frac{1}{5} \text{ and } p(A=1,B=1) = \frac{4+\eps}{10+4\eps} \approx \frac{2}{5}.
\end{equation}
Moreover, it turns out that $\gamma^4 := \Infc[p]{R}{X^iY^i}{TW,A=0,B=0}$ is equal to $\Infc[q]{R}{X^iY^i}{TW,\CD_i}$. We view this as an embedding of the AND function: if $\mathrm{AND}(A,B)=0$, then $\zeros(X,Y) = T_u \cup T_0$ while if $\mathrm{AND}(A,B)=1$, then because of the choice of a different window, the edge labels of the path get flipped  along the path connecting the terminals in $T_i$ and $\zeros(X,Y) = T_u \cup T_0 \cup T_i$.

The crux of the proof is to show that if $\gamma \ll 1$, then the probability of the two events in \eqref{eqn83gdhdls} must be quite different. To see this, call a triple $(r,t,w)$ good if the distributions $p(A|R=r,T=t,W=w)$ and $p(B|R=r,T=t,W=w)$ are both close to the uniform distribution on a bit and denote by $\CG$ the event that the triple is good.
Using Pinsker's inequality and properties of the gadget, we can show that most of the contribution comes from good triples (see \Cref{bigprop1}), i.e., $p(A = 0, B = 0) \leq p(A = 0, B = 0, \CG) + O(\gamma)$.

We will see that, for any good $(r,t,w)$, the conditional distribution $p(A,B|R=r,T=t,W=w)$ is very close to the uniform distribution on two bits. This by itself does not give us a contradiction to \eqref{eqWednesday} as the above inequality does not imply that the probability $p(\CG)$ is large (it only shows that the probability $p(\CG\mid A=0, B=0)$ is large).
So, we further partition the good triples $\CG$ into $\CG_1$ and $\CG_2$.

\paragraph{Two-Intersecting Family.} The triples in $\CG_2$ correspond to an event where all the good pairings (after fixing all but the pairing of the 5 nodes in $T_u \cup T_0 \cup T_i$) form a $2$-intersecting family of $\binom{[5]}{3}$. We use bounds on the size of intersecting families to show that $p(A=0,B=0,\CG_2)$ is roughly $\frac{4}{\binom{5}{3}} = \frac25$ times $p(A=1,B=1)$ (see \Cref{bigprop2}).

\paragraph{Triples Containing Small Entries.} To deal with the remaining triples in $\CG_1$, we first note that any pairing $t$ (along with $w$) chooses a rectangular submatrix $\CA_t \times \CB_t$ of the slack matrix $S$ where $\CA_t$ and $\CB_t$ denote the set of rows and columns respectively. We show that for any good triple $t$ (along with $r,w$) one can find another good $t'$ (along with $r,w$) such that the two rectangles intersect and moreover, the rectangle $\CA_{t'} \times \CB_{t}$ contains entries where $|\zeros(x,y)|=1$. In light of \eqref{eqWednesday}, the probability of the event $|\zeros(X,Y)|=1$ in the original unconditioned probability space, denoted $\alpha$, is $O(\eps)$. We are able to show that the probability contribution $p(A=0,B=0,\CG_1)$ can be bounded by a constant factor of $\alpha$ (see \Cref{bigprop3}). We remark that if $\eps=0$, this fact is much simpler to prove as $\alpha = 0$ in this case, and one could even afford a loose bound of $2^n \alpha$ for instance. We, however, need a very precise quantitative bound here which increases the complexity of the arguments. 

Overall, the above implies
\[ p(A = 0, B = 0) \leq \frac{2}{5} \cdot  p(A = 1, B = 1) + O(\eps + \gamma) = \frac{4}{25} + O(\eps + \gamma).\]
If $\gamma = O(\eps)$, taking $\eps$ to be a small enough constant, the right hand side above is strictly smaller than $1/5$, which contradicts the true probability given by \eqref{eqn83gdhdls}.

We note that the proof for the matching polytope proceeds along very similar lines (see \cite[Ch. 12]{RY20}), but is somewhat simpler compared to the present proof. The difficulty in the present proof is primarily due to the fact that the edge-disjoint paths above depend on the pairing of all the terminals. 
\section{Preliminaries}
\label{secPrelim}

\subsection{Notations for probability spaces}
All random variables considered here are discrete and are denoted by capital letters (e.g.\ $A$), and values they attain are denoted by lower-case letters (e.g.\ $a$). Events in a probability space will be denoted by calligraphic letters (e.g.\ $\CE$). For events $\CD$ and $\CE$, we use $\comp{\CD}$ to denote the complement and $\CD,\CE$ to denote the intersection $\CD \cap \CE$.

Given a probability space with probability measure $p$ and a random variable $A$ defined on the underlying sample space, we use the notation $p(A)$ to denote the distribution of the variable $A$ with respect to the probability measure $p$.
Given an event $\CW$ in a probability space $p$, we will denote by $p(\cdot\mid\CW)$, the probability space $p$ after conditioning on the event $\CW$.
The term $p(\CD)$ denotes the probability of the event $\CD$ according to $p$.

For conciseness, we use $p(ABC) = p(A, B, C)$ to denote the distribution of multiple random variables, and denote by $p(abc) = p(a, b, c)$ the probability of values $a, b, c$ according to the distribution $p(ABC)$.
The event $A = a$ is often simply described as $a$, i.e. $p(A = a) = p(a)$ or $p(B \mid A = a) = p(B \mid a)$.

The support of a random variable $A$ is defined to be the set $\supp(A) := \{a \mid p(a)>0\}$.
Given a fixed value $a$ of $A$, we denote the expected value of a function $f(a)$ under the distribution $p(A)$ by $\Ex{p(a)}{f(a)} := \sum_{a} p(a) \cdot f(a)$.

We say that a random variable $A$ determines another random variable $B$ if there is a function $h$ such that $h(A)=B$. We write $A-M-B$ in the probability space $p(\cdot)$ if $A$ and $B$ are independent given $M$, i.e., $p(amb) = p(m) \cdot p(a|m) \cdot p(b|m)$ holds for every $a,b,m$.
In this case, we say that $A$, $M$, and $B$ form a \emph{Markov chain}. 

\subsection{Information theory basics}

In this section, we mention the information theoretic tools that we need for our proof.
For a random variable $A$ over a probability space with distribution $p$, the \emph{entropy} of $A$ (with respect to $p$) is defined as
\[
    \BH_p(A) = \BE_{\cp(a)}\left[\log_2\frac{1}{\cp(a)}\right].
\]
Whenever $p$ is clear from the context, we simply write $\BH(A) = \BH_p(A)$.
For any two random variables $A$ and $B$, the entropy of $A$ conditioned on $B$ is defined as $\BH(A|B) = \BE_{\cp(b)}[\BH(A|b)]$.
Given an event $\CW$ and random variables $A,B$ in a probability space with distribution $p$, we write $\BH_p(A|B,\CE) = \BH_q(A|B)$ where $q = p(\cdot \mid \CW)$ arises from $p$ by conditioning on the event $\CW$.

The \emph{mutual information} between two random variables $A,B$ over a probability space with distribution $p$ is defined as
\[
    \Infp{A}{B}{p} = \BH_p(A) - \BH_p(A|B).
\]
Again, we omit the subscript if $p$ is clear from the context.
Note that since $\BH_p(A|B) = \BH_p(B|A)$ we have $\Inf{A}{B} = \Inf{B}{A}$.
The \emph{conditional mutual information} is defined as $\Infc{A}{B}{C} = \BH(A|C) - \BH(A|BC)$.
Given an event $\CW$, we define $\Infc[p]{A}{B}{C,\CW} = \Infc[q]{A}{B}{C}$, where $q = p(\cdot \mid \CW)$.

The following basic facts will be needed throughout our proof.
Some proofs can be found in the book \cite{CT06} by Cover \& Thomas.
Below $A,B,C,\ldots$ are arbitrary random variables and $\CE$ an arbitrary event in a probability space $p(\cdot)$ unless explicitly mentioned otherwise.

\begin{proposition}[{\cite[Thm.~2.6.5]{CT06}}]
    $\BH(A|B) \le \BH(A)$ where the equality holds if and only if $A$ and $B$ are independent.
\end{proposition}


\begin{proposition} \label{prop:chainrule}
   If the random variables $A_1,...A_n$ are mutually independent, then $\Infc{A_1,\dots, A_n}{B}{C} \ge \sum_{i=1}^n \Infc{A_i}{B}{C}$.
\end{proposition}

\begin{proposition} \label{prop:infoexpected}
    $\Infc{A}{B}{CD} = \Ex{p(d)}{\Infc{A}{B}{C, d}}$. In particular, if $D$ is determined by $C$, then $\Infc{A}{B}{C} = \Ex{p(d)}{\Infc{A}{B}{C, d}}$.
\end{proposition}

The above implies that if $A$ and $B$ are independent then $\Inf{A}{B}=0$.
In the next statement, we refer to the \emph{statistical distance} between two distributions $p(A)$ and $q(A)$, which is defined as $|p(A) - q(A)| = |q(A) - p(A)| = \frac12 \sum_{a} |p(a) - q(a)|$.
Note that the statistical distance satisfies the triangle inequality.

\begin{proposition}[Pinsker's inequality] \label{proposition:pinsker}
    For any random variables $A, B, C$, we have
        $\BE_{p(bc)}|p(A|bc)- p(A|c)| \le \sqrt{\Infc{A}{B}{C}}.$
    
\end{proposition}

We will also make use of the following basic facts.

\begin{proposition} \label{prop:bitcond} If $A \in \bits$ is uniform, then  $|p(R|A=0)-p(R|A=1)| = \BE_{p(r)}|p(A=0|r)-p(A=1|r)|$.
\end{proposition}
\begin{proof}
    Using Bayes' rule, we have $p(r|A=0) = \frac{p(r)p(A=0|r)}{p(A=0)} = 2p(r)p(A=0|r)$. Therefore, 
    \begin{multline*}
        |p(R|A=0)-p(R|A=1)| = \frac12 \sum_{r} |p(r|A=0)-p(r|A=1)| \\
        = \BE_{p(r)}|p(A=0|r)-p(A=1|r)|. \quad \qed
    \end{multline*}
\end{proof}

\begin{proposition} \label{prop:markov2}
    If $R - X - W$, then $\BE_{p(w)}|p(R|w) - p(R)| \le \BE_{p(x)}|p(R|x) - p(R)|$.
\end{proposition}
\begin{proof}
    Since $W - X - R$ is a Markov chain, we have $\BE_{p(w)}|p(R|w) - p(R)| =\BE_{p(w)}\big|\BE_{p(x|w)}[p(R|wx)] - p(R)\big| \le\BE_{p(wx)}\big|\BE_{p(R|wx)} - p(R)\big|$
    .
    The proof follows since $ \BE_{p(wx)}\big|\BE_{p(R|wx)} - p(R)\big| = \BE_{p(x)}|p(R|x)-p(R)|$.
    \qed
\end{proof}

\begin{proposition} \label{prop:markov1}
    If $R - E - W$ and $R - X - EW$, then we have $\Infc{R}{X}{E} \le \Infc{R}{X}{W}$.
\end{proposition}
\begin{proof}
    As $\BH(R|E) = \BH(R|EW) \le \BH(R|W)$ and $\BH(R|XE) = \BH(R|XEW) = \BH(R|XW)$, we conclude
    $\Infc{R}{X}{E} = \BH(R|{E}) - \BH({R}|{XE}) \le \BH({R}|{W}) - \BH({R}|{XW}) = \Infc{R}{X}{W}$.
    \qed
\end{proof}

\subsection{Properties of the Gadget Function}

The gadget function $g: \CX \times \CY \to \bits$ (with alphabet $\CX=\CY=\bits^3$) has crucial properties that we will use throughout the proof and collect here.
To this end, let us recall the notion of a \emph{window}, which is a set $w = \{(x,y), \, (x',y')\}$ of two distinct elements from $\CX \times \CY$ such that $x = x'$ (horizontal) or $y = y'$ (vertical) and $g(x,y)=g(x',y')$, i.e. it is a $2 \times 1$ (horizontal) or a $1\times 2$ (vertical) rectangle where the value of the gadget is constant, see \Cref{figWindows}.
Given a (horizontal) window $w = \{(x,y), \, (x,y')\}$, we associate the four sets 
$
    \ps{w}{00}  = \{(\tilde{x},\tilde{y}), \, (\tilde{x},\tilde{y}')\},
    \ps{w}{10} = \{(x,\tilde{y}), \, (x,\tilde{y}')\},
    \ps{w}{01} = \{(\tilde{x},y), \, (\tilde{x},y')\},
    \text{ and }
    \ps{w}{11} = w,
$
where
\begin{align*}
    \tilde{x} & = (1 + x_1 + y_2y_3' + y_2'y_3, x_2 + y_3 + y_3', x_3 + y_2 + y_2') \pmod{2} \\
    \tilde{y} & = (1 + y_1 + x_2 + x_3(1 + y_3 + y_3'), y_2 + 1, y_3' + 1) \pmod{2} \\
    \tilde{y}' & = (1 + y_1' + x_2 + x_3(1 + y_3 + y_3'), y_2' + 1, y_3 + 1) \pmod{2}.
\end{align*}

Since $g$ is symmetric with respect to its inputs $x,y$, we may also define $\ps{w}{00}$, $\ps{w}{10}$, $\ps{w}{01}$, and $\ps{w}{11}$ analogously for (vertical) windows $w = \{(x,y), \, (x',y)\}$ by changing the roles of $x,y$.
We note that $\ps{w}{00}$, $\ps{w}{10}$, $\ps{w}{01}$, and $\ps{w}{11}$ are pairwise distinct.
If $w = \{(x,y), \, (x',y')\}$ is a window with $g(x,y) = g(x',y') = b$, we say that $w$ is a \emph{$b$-window}. 
Defining $\comp{b} = 1 - b$, we observe the following.

\begin{repclm}{clmwindowa}
    If $w$ is a $b$-window, then $\ps{w}{00}$, $\ps{w}{10}$, and $\ps{w}{01}$ are $\comp{b}$-windows.
\end{repclm}

The proofs of this claim and all subsequent claims only require a straightforward calculation or may be even clear to the advanced reader.
However, we give detailed proofs of all claims in \Cref{secClaims}.

In view of the above claim, we may view $\ps{w}{00}$, $\ps{w}{10}$, $\ps{w}{01}$, $\ps{w}{11}$ as a ``stretched'' \textsf{AND} if $w$ is a $1$-window, or a ``stretched'' \textsf{NAND} if $w$ is a $0$-window, see Figure~\ref{figWindows} for an illustration.
Another observation that we need is the following:

\begin{repclm}{clmwindowb}
    For each $i,j \in \{0,1\}$, the map $w \mapsto \ps{w}{ij}$ is a bijection on the set of windows.
\end{repclm}

\subsection{Superadditivity}
\label{secSuperadditivity}

Our main goal is to prove \Cref{thm99dhsuao}, i.e., there is some $\eps > 0$ such that $\Infc[q]{R}{XY}{TW,\CD} = \Omega(k)$.
To this end, we exploit the superadditivity property of mutual information.
In particular, for $i \in [k]$, we shall show that $R$ gives $\Omega(1)$ bits of information about the edge labels for each of the $k$ edge disjoint paths between the terminals in $T_1, \cdots T_k$, which when combined with the superadditivity property implies that $R$ gives $\Omega(k)$ bits of information about $X,Y$ giving us \Cref{thm99dhsuao}.
To this end, we first make the following observation.

\begin{repclm}{clmweee}
    The variables $X_eY_e$ for $e \in E$ are mutually independent in the probability space $q(\cdot \mid tw,\CD)$.
\end{repclm}

Using this claim, \Cref{prop:chainrule}, and nonnegativity of mutual information we obtain
\begin{equation}
    \label{eqhd9s02k3}
    \Infc[q]{R}{XY}{TW,\CD} \ge \sum_{i=1}^{k} \Infc[q]{R}{X^iY^i}{TW,\CD}.
\end{equation}

Next, for each $i \in [k]$, recall the event
\[
    \CD_i: \zeros(X, Y) = T_u \cup T_0 \text{ and } (X, Y) \in \ps{W^i}{00}W^{-i},
\]
where $\ps{W^i}{ab}W^{-i}$ arises from $W$ by replacing each entry $w$ that corresponds to an edge in $\Path_i(T)$ by $\ps{w}{ab}$.
For conciseness of notation, we denote $T_u \cup T_0$ by $T_*$.
Note that if $\zeros(\ps{W^i}{00}W^{-i}) = T_*$, then $\zeros(\ps{W^i}{01}W^{-i}) = \zeros(\ps{W^i}{10}W^{-i}) = T_*$ as well as $\zeros(W) = T_* \cup T_i$.
Although the event $\CD_i$ can be considered a part of this smaller rectangle, note that $\CD_i$ is a not a subset of $\CD$.
However, since $W \mapsto \ps{W^i}{00}W^{-i}$ is a bijection over the space of windows, the mutual information quantities remain the same:

\begin{repclm}{clmblub}
    $ \Infc[q]{R}{X^iY^i}{T W, \CD} = \Infc[q]{R}{X^iY^i}{TW,\CD_i}$.
\end{repclm}

Thus, by \eqref{eqhd9s02k3} we obtain $\Infc[q]{R}{XY}{TW,\CD} \ge \sum_{i=1}^{k} \Infc[q]{R}{X^iY^i}{TW,\CD_i}$, and hence it indeed suffices to prove \Cref{thmLowerBoundDi}, which states that, for $\eps > 0$ small enough, there is a constant $c > 0$ such that $ \Infc[q]{R}{X^iY^i}{TW,\CD_i} \geq c$ holds for every $i \in [k]$.
The proof of \Cref{thmLowerBoundDi} is fairly involved and presented in the next section.
\section{Proof of \Cref{thmLowerBoundDi}}
\label{secProofKeyThm}

This section is devoted to the proof of \Cref{thmLowerBoundDi}, which claims that there is a constant $c > 0$ such that $ \Infc[q]{R}{X^iY^i}{TW,\CD_i} \geq c$ holds for every $i \in [k]$, provided that $\eps > 0$ is small enough.
To this end, we consider the event $\CD_i$ as embedded in a smaller rectangle arising from $\ps{W^i}{ab}W^{-i}$ for $a,b \in \bits$ as described before. Recall that $\CD_i$ is the event that $\Zeros(\ps{W^i}{00}W^{-i})=T_*$, which also implies $\Zeros(\ps{W^i}{01}W^{-i})=\Zeros(\ps{W^i}{10}W^{-i})=T_*$ and $\Zeros(W) = T_* \cup T_i$. We now fix some part of the probability space.

Let us define $Z = \Zeros(W)$ and $T_{-i} = T_1 \dots T_{i - 1} T_{i + 1} \dots T_k$. Note that this deviates from our definition of $W^{-i}$ in the sense that $T_{-i}$ does not include $T_0$. Fix any $i \in [k]$ and $zt_{-i}$ such that $q(zt_{-i} \mid \CD_i) > 0$.
\begin{repclm}{clmztuior}
    $z = \CT \setminus t_{-i}$, where $\CT$ is the set of all terminals.
\end{repclm}




    Defining $\gamma = \sqrt[4]{\Infc[q]{R}{X^iY^i}{TW, zt_{-i} \CD_i}}$, we will show that $\gamma$ can be bounded from below in terms of $\eps$ only, and for $\eps > 0$ small enough, this lower bound will be a positive universal constant (that is independent of $G,i,z,t_{-i}$).
    Together with \Cref{prop:infoexpected}, this implies the statement of \Cref{thmLowerBoundDi}.

    To prove the above, let us consider the event
    \[
    \CE: \exists a, b \in \bits \text{ with } X^iY^i \in \ps{W^i}{ab} \text{ and } X^{-i}Y^{-i} \in W^{-i},
    \]
    which is a rectangle. Define random variables $A,B \in \bits^2$ to be the indicator random variables for the event that $X^iY^i \in \ps{W^{i}}{ab}$. We can now consider a new smaller probability space $p(\cdot) = q(\cdot \mid zt_{-i},\CE)$. Note that the event $\CD_i$ is the event that $A=0, B=0$ in this probability space. One should view the above as an embedding of the  two bit AND function: if $\mathrm{AND}(a,b)=0$, then $\zeros(X,Y) = T_*$ while if $\mathrm{AND}(a,b)=1$, then the edge labels of the path get flipped  along the path connecting the terminals in $T_i$ and $\zeros(X,Y) = T_* \cup T_i$.
    This implies that the slack entry $S(X,Y)$ for any $X,Y \in \ps{W^i}{ab} W^{-i}$ equals ${2+\eps}$ if $a=0$ or $b=0$ and equals $4+\eps$ if both $a=b=1$ and the probabilities in the space $p(\cdot)$ are obtained by normalizing these slack matrix entries. In particular, this will imply:

    \begin{proposition}
        \label{bigprop0}
        We have $p(A = 0, B = 0) = \frac{2 + \eps}{10 + 4\eps}$ and
        $p(A = 1, B = 1) = \frac{4 + \eps}{10 + 4\eps}$.
    \end{proposition}
Our goal is to show that $\gamma$ is polynomially related to $\eps$ in order to achieve these probabilities. To do this, we note that the marginal distribution over $A$ and $B$ in this probability space is that of a uniform random bit, as can be seen by considering the slack matrix entries. Moreover, conditioned on $R,T,W$, it turns out that the random variables $A$ and $B$ are independent in this probability space.  

Thus if $\gamma \ll 1$, then Pinsker's inequality and properties of the gadget will imply that conditioning on a typical value of $R=r, T=t, W=w$ does not change the distribution of $A$ and $B$ much, i.e. conditioned on a typical $(r,t,w)$, the random variables $A$ and $B$ are distributed close to two uniform random bits. To formalize this, we define ${\alpha}_{rtw} = |p(A=0|rtw) - p(A=1|rtw)|$ and $\beta_{rtw} = |p(B=0|rtw) - p(B=1|rtw)|$. We say that $(r,t_*,w)$ is good if for any $t$ that refines $t_* \sqcup t_i$, we have $\alpha_{rtw} \le \gamma$ and $\beta_{rtw} \le \gamma$. Let $\CG$ denote the event that the triple $(R, T_*, W)$ is good. We will show  that most of the contribution to the event $A=0, B=0$ comes from the good triples:
    \begin{proposition}
        \label{bigprop1}
        We have $p(A = 0, B = 0) \leq p(A = 0, B = 0, \CG) + O(\gamma)$.
    \end{proposition}

Note that for any good $(r,t,w)$, we have that for any $a,b \in \bits^2$, 
\[ p(A=a, B=b \mid rtw)  = 1/4 \pm \Theta(\gamma),\]
which is promising. This by itself does not give us any contradiction as the above proposition does not imply that the probability $p(\CG)$ is large (it only shows that the probability $p(\CG\mid A=0, B=0)$ is large). So, we further partition the good triples $\CG$ into $\CG_1$ and $\CG_2$, where $\CG_2$ corresponds to all good $(r, t_*, w)$ such that all $t_*'$ for which $(r, t_*', w)$ is good satisfy $|t_* \cap t_*' | \geq 2$, and $\CG_1 = \CG \setminus \CG_2$.

Note that there are $\binom{5}{3}$ choices of splitting $t = t_* \cup t_i$, but in $\CG_2$, if we fix an $r,w$, then the set $\{ t \mid (r,t,w) \in \CG_2\}$ forms a 2-intersecting family and Erdos-Ko-Rado theorem implies that the size of this family is at most $4$. Moreover, for any $(r,t,w) \in \CG_2$ the event where $A=1,B=1$ is equivalent to the event $\Zeros(X, Y) = T_* \cup T_i$ in the probability space $p(\cdot)$. In particular, we shall use this and the properties of the gadget to show that the contribution to the event $A=B=0$ that comes from $\CG_2$ is roughly $\frac{4}{\binom{5}{3}} = \frac25$ fraction of  the probability of the event $p(A=1,B=1)$:
    \begin{proposition}
        \label{bigprop2}
        If $\gamma$ is small enough, then $p(A = 0, B = 0, \CG_2) \leq \frac{2}{5} \cdot p(A = 1, B = 1) + O(\gamma)$.
    \end{proposition}

Now consider the set $\CG_1$ of remaining good triples. For any $(r,t,w)\in \CG_1$, there is another good $(r,t',w) \in \CG$ such that $|t\cap t'|=1$. In this case, we shall see that the two rectangles given by the partitions $t$ and $t'$ intersect and their intersection gives rise to slack matrix entries with value $\eps$. The probability of this event in the unconditioned probability space is $O(\eps)$. We are able to relate the contribution of the event that $p(A=0,B=0,\CG_1)$ to within constant factor of the above event. 

    \begin{proposition}
        \label{bigprop3}
        If $\gamma$ is small enough, then
        $p(A = 0, B = 0, \CG_1) = O(\eps)$.
    \end{proposition}

Note that in the exact case where $\eps = 0$, the slack matrix entry is zero, so one can be pretty loose with relating the probability of the above event under $p$ to the the probability in the unconditioned probability space which is also zero. Thus, the above is a lot easier to prove for the exact case. However, here since we want to relate the terms within a constant factor and $\eps$ is also a constant, this requires a more careful computation and the proof becomes more involved.

Given the above, if $\gamma$ is small enough, then using \Cref{bigprop0}, \Cref{bigprop1}, \Cref{bigprop2}, and \Cref{bigprop3} we obtain
    \begin{align*}
         \tfrac{2+\eps}{10+4\eps} &= p(A = 0, B = 0) \leq p(A = 0, B = 0, \CG) + O(\gamma) \\
        &= p(A = 0, B = 0, \CG_1) + p(A = 0, B = 0, \CG_2) + O(\gamma) \\
        &\leq p(A = 0, B = 0, \CG_1) + \tfrac25 p(A = 1, B = 1) + O(\gamma) \\
        &= O(\eps) + \tfrac{4 + \eps}{25 + 10\eps} + O(\gamma).
    \end{align*}
    By rearranging, we obtain $\frac{2 + 3\eps}{50+20\eps} \le O(\eps) + O(\gamma)$.
    Therefore, if $\eps > 0$ is small enough, we obtain a constant positive lower bound on $\gamma$.
    Thus, it remains to prove Propositions~\ref{bigprop0}--\ref{bigprop3}.



\subsection{Proof of \Cref{bigprop0}}
The proof follows from the fact that the probability space is obtained by normalizing the slack matrix entries, which are $2+\eps$ if $a=0$ or $b=0$, and $4+\eps$ otherwise.
To this end, we first observe that for every $t,w,a,b$ we have
\begin{align*}
    p(A = a, B = b, tw) &= q(zt_{-i}, \CE)^{-1} \cdot q(A = a, B = b, tw, zt_{-i}, \CE) \\
    &= q(zt_{-i}, \CE)^{-1} \cdot q\left(tw, XY \in \ps{w^i}{ab}w^{-i}\right) \\
    &= q(zt_{-i}, \CE)^{-1} \cdot q(tw) \cdot q \left(XY \in \ps{w^i}{ab}w^{-i}\right) \\
    &= \frac{q(tw) \cdot 2^{|E|} }{q(zt_{-i}, \CE) \cdot ||S||_1} \cdot \left( \left| \Zeros(\ps{w^i}{ab}w^{-i}) \right| - 1 + \eps \right).
\end{align*}
Furthermore, if $t_u, t_0, t_i, w$ satisfy $p(tw) > 0$, then $\Zeros(w) = z = t_u \cup t_0 \cup t_i$ by Claim \ref{clmztuior},
and hence we have $\Zeros(\ps{w^i}{00}w^{-i}) = \Zeros(\ps{w^i}{01}w^{-i}) = \Zeros(\ps{w^i}{10}w^{-i}) = t_u \cup t_0$ as well as $\Zeros(\ps{w^i}{11}w^{-i}) = t_u \cup t_0 \cup t_i$.
Thus, we obtain
\begin{align*}
    p(A=0,B=0)
    &= \sum_{tw : p(tw) > 0} \frac{p(A=0,B=0,tw)}{\sum_{a,b \in \bits} p(A=a,B=b,tw)} \\
    &= \frac{2+\eps}{3(2+\eps)+(4+\eps)}
    = \frac{2+\eps}{10+4\eps}
\end{align*}
and
\begin{align*}
    p(A=1,B=1)
    &= \sum_{tw : p(tw) > 0} \frac{p(A=1,B=1,tw)}{\sum_{a,b \in \bits} p(A=a,B=b,tw)} \\
    &= \frac{4+\eps}{3(2+\eps)+(4+\eps)}
    = \frac{4+\eps}{10+4\eps}.
\end{align*}

\subsection{Proof of \Cref{bigprop1}}

Since
  \begin{align*}
      p(A = 0, B = 0) &= p(A = 0, B = 0, \CG) + p(A = 0, B = 0, \comp{\CG}) \\
      &= p(A = 0, B = 0, \CG) + p(\comp{\CG} \mid A = 0, B = 0) \cdot p(A = 0, B = 0) \\
      &\leq p(A = 0, B = 0, \CG) + p(\comp{\CG} \mid A = 0, B = 0),
  \end{align*}
    it suffices to show that
    $p(\comp{\CG} \mid A = 0, B = 0) = O(\gamma)$ holds.
    
    To this end, let us consider three (correlated) ways of generating all refinements $T, T', T''$ that agree with the split $T_* \sqcup T_i$. We identify the nodes in $T_*$ with $\{1,2,3\}$ and let $U$ be a uniformly random element of $\{1,2,3\}$. Let the unpaired nodes in the pairings $T, T', T''$ be $T_u := U$, $T'_u = U + 1 \mod 3$, and $T''_u = U + 2 \mod 3$, respectively. Note that the marginal distribution of $T_u$ is the same as that of $T_u'$ and $T_u''$.
    Next, let $\CB_\alpha = \{ rtw \mid \alpha_{rtw} > \gamma\}$, define $\CB_\beta$ analogously, and set $\CB = \CB_\alpha \cup \CB_\beta$.
    Note that if $rt_*w \in \comp{\CG}$, then either $rtw \in \CB$ or $rt'w \in \CB$ or $rt''w \in \CB$. As they have the same marginal distributions, by a union bound, we obtain $p(\comp{\CG} \mid A=0,B=0) \le 3 \cdot p(rtw \in \CB \mid A=0,B=0)$. Again by a union bound, we have $p(\CB \mid A = 0, B = 0) \leq p(\CB_\alpha \mid A = 0, B = 0) + p(\CB_\beta \mid A = 0, B = 0)$. In what follows, we will show that $p(rtw \in \CB_\alpha \mid A=0,B=0) = O(\gamma)$ holds. $p(rtw \in \CB_\beta \mid A=0,B=0) = O(\gamma)$ follows analogously.
    To this end, we make use of the following claim which again follows from the fact that the probability space (without conditioning on $r$) is obtained from normalizing the slack matrix entries :

    \begin{repclm}{clmwednf}
        $p(A = 0, B = 0, tw) = p(A = 1, B = 0, tw)$.
    \end{repclm}
    In particular, this implies $p(A=0 \mid B=0) = p(A=1 \mid B=0) = \frac12$.
    Thus, by Markov's inequality, we have
    \begin{multline*}
        p(\CB_\alpha \mid A = 0, B = 0)
        = p(\CB_\alpha,A=0 \mid B = 0) \cdot p(A=0 \mid B=0)^{-1} \\
        = 2 \cdot p(\CB_\alpha,A=0 \mid B = 0) 
        \leq 2 \cdot p(\CB_\alpha \mid B = 0)
        \leq \frac{2}{\gamma} \BE_{p(rtw|B=0)}[\alpha_{rtw}].
    \end{multline*}
    Hence, it remains to show that $\BE_{p(rtw | B = 0)}[\alpha_{rtw}] = O(\gamma^2)$ holds.
    To this end, note that:
    \begin{repclm}{clmuoetzwh}
        $A - R - B$ in $p(\cdot \mid tw)$.
    \end{repclm}
    Using Claim \ref{clmuoetzwh} and \Cref{prop:bitcond}, which we can apply since Claim \ref{clmwednf} implies that $p(A \mid tw, B = 0)$ is uniform, we obtain
    \begin{align*}
        &\quad \BE_{p(rtw|B=0)}[\alpha_{rtw}] \\
        &= \BE_{p(rtw \mid B = 0)} |p(A = 0 \mid rtw) - p(A = 1 \mid rtw)|\\
        &= \BE_{p(rtw \mid B = 0)} |p(A = 0 \mid rtw, B = 0) - p(A = 1 \mid rtw, B = 0)| \\
        &= \BE_{p(tw \mid B = 0)} \BE_{p(r \mid tw, B = 0)} |p(A = 0 \mid rtw, B = 0) - p(A = 1 \mid rtw, B = 0) | \\
        &= \BE_{p(tw \mid B = 0)} |p(R \mid tw, A = 0, B = 0) - p(R \mid tw, A = 1, B = 0)|.
    \end{align*}
    Thus, it remains to show that $\BE_{p(tw \mid B = 0)} |p(R \mid tw, A = 0, B = 0) - p(R \mid tw, A = 1, B = 0)| = O(\gamma^2)$ holds.
    For this, we observe:
    \begin{repclm}{zwezroi}
        For any $a \in \bits$, we have $p(TW \mid A = a, B = 0) = p(TW \mid B = 0)$ and $p(R \mid tw^{-i}, A = a, B = 0) = p(R \mid tw^{-i}, B = 0)$.
    \end{repclm}
    
    Thus, we have
    \begin{align*}
        &\quad \BE_{p(tw \mid B = 0)} |p(R \mid tw, A=0, B=0) - p(R \mid tw, A = 1, B = 0)| \\
        &\leq \BE_{p(tw \mid B = 0)} \bigl[ |p(R \mid tw, A=0, B=0) - p(R \mid tw^{-i}, B=0)|\\
        & \quad+ |p(R \mid tw, A = 1, B=0) - p(R \mid tw^{-i}, B = 0)| \bigr] \\
        &= \BE_{p(tw \mid A = 0, B = 0)} |p(R \mid tw, A = 0, B = 0) - p(R \mid tw^{-i}, A = 0, B = 0)| \\
        &\quad + \BE_{p(tw \mid A =1, B=0)} |p(R \mid tw, A = 1, B = 0) - p(R \mid tw^{-i}, A = 1, B = 0)|,
    \end{align*}
    where the inequality follows from the triangle inequality for statistical distance.
    Hence, it suffices to show that $\BE_{p(tw \mid a, B = 0)} |p(R \mid twa, B = 0) - p(R \mid tw^{-i}a, B = 0)| = O(\gamma^2)$ holds for any $a \in \bits$.
    In order to do this we make use of the following lemma which bounds the total variation distance to the mutual information. This lemma relies on Pinsker's inequality and properties of the gadget and will be proven later.
    \begin{lemma}
        \label{lem:graphtheory}
        We have
        \begin{align*}
            &\BE_{p(w^i \mid tw^{-i}a, B = 0)} |p(R \mid twa, B = 0) - p(R \mid tw^{-i}a, B = 0)| \\
            = &O \left(\sqrt{\Infc[p]{R}{X^iY^i}{W^i, tw^{-i}a, B = 0}} \right).
        \end{align*}
    \end{lemma}
    
    Using this lemma, we have
    \begin{align*}
        &\BE_{p(tw \mid a, B = 0)} |p(R \mid twa, B = 0) - p(R \mid tw^{-i}a, B = 0)| \\
        &= \BE_{p(tw^{-i} \mid a, B = 0)} \left[ \BE_{p(w^i \mid tw^{-i}a, B = 0)} |p(R \mid twa, B = 0) - p(R \mid tw^{-i}a, B = 0)|\right] \\
        &= \BE_{p(tw^{-i} \mid a, B = 0)} \left[ O \left( \sqrt{\Infc[p]{R}{X^iY^i}{W^i, tw^{-i}a, B = 0}} \right) \right]\\
        &= O \left( \sqrt{\BE_{p(tw^{-i} \mid a, B = 0)} \Infc[p]{R}{X^iY^i}{W^i, tw^{-i}a, B = 0}} \right) \\
        &= O \left( \sqrt{\Infc[p]{R}{X^iY^i}{TW, a, B = 0}} \right),
    \end{align*}
    where the second to last equality follows from Jensen's inequality and the concavity of $\sqrt{\cdot}$, and the last equality is implied by \Cref{prop:infoexpected}. Thus, it is left to show that $\Infc[p]{R}{X^iY^i}{TW, a, B = 0} = \gamma^4$ holds.
    To show this, consider the following observation, which follows from the fact that the marginal distribution on windows $\ps{W^i}{01}W^{-i}$ is the same as that of $\ps{W^i}{00}W^{-i}$.
    \begin{repclm}{jjiofe}
        $\Infc[p]{R}{X^iY^i}{TW, A = 1, B = 0} = \Infc[p]{R}{X^iY^i}{TW, A = 0, B = 0}$.
    \end{repclm}
    This implies
    \begin{align*}
        \Infc[p]{R}{X^iY^i}{TW,a, B = 0} &= \Infc[p]{R}{X^iY^i}{TW, A = 0, B = 0} \\
        &= \Infc[q]{R}{X^iY^i}{TW, A = 0, B = 0, zt_{-i}, \CE}.
    \end{align*}
    Recall that for $A = B = 0$, the event $\CE$ is equivalent to $XY \in \ps{W^i}{00} W^{-i}$ which further implies $\Zeros(X, Y) = T_u \cup T_0$ since $z = T_u \cup T_0 \cup T_i$ by Claim \ref{clmztuior}. This shows 
    \[
        \Infc[q]{R}{X^iY^i}{TW, A = 0, B = 0, zt_{-i}, \CE}
        = \Infc[q]{R}{X^iY^i}{TW, zt_{-i}, \CD_i} = \gamma^4.
    \]
    Thus, it remains to prove \Cref{lem:graphtheory}.

\begin{proof}[Proof of \Cref{lem:graphtheory}]
Let us fix any $a \in \bits$ and $tw^{-i}$. For the sake of notation, we denote $p(\cdot \mid tw^{-i}a, B = 0)$ by $s(\cdot)$.
Let $\tilde{X}^i\tilde{Y}^i$ be a uniformly random input in $\ps{W^i}{a 0}$ such that $\tilde{X}^i\tilde{Y}^i$ and $X^iY^iR$ are independent in $s( \cdot \mid w^i)$ for all windows $w^i$.
We start with the following observations.

\begin{repclm}{hkfdsua}
    There exists a $0/1$-vector $b^i$ such that every window $w^i$ with $s(w^i) > 0$ is a $\comp{b^i}$-window.
\end{repclm}

\begin{repclm}{clmhkfdsud}
    In $s(\cdot)$, we have $R - X^iY^i - \{X^iY^i,\tilde{X}^i\tilde{Y}^i\}W^i$.
\end{repclm}

\begin{repclm}{clmhtoewph}
    In $s(\cdot)$, we have $R - \{X^iY^i,\tilde{X}^i\tilde{Y}^i\} - W^i$.
\end{repclm}

By Claim \ref{clmhkfdsud}, Claim \ref{clmhtoewph} and \Cref{prop:markov1}, we have
\begin{equation}
    \label{eqn:einsteadofw}
    \Infc[s]{R}{X^iY^i}{\{X^iY^i,\tilde{X}^i\tilde{Y}^i\}} \le \Infc[s]{R}{X^iY^i}{W^i}.
\end{equation} 
We will show that
\begin{equation}
    \label{eqn:graphthings}
    \BE_{s(x^iy^i)} |s(R \mid x^iy^i) - s(R)| = O \left( \sqrt{\Infc[s]{R}{X^iY^i}{\{X^iY^i,\tilde{X}^i\tilde{Y}^i\}}} \right)
\end{equation}
holds.
Together with \eqref{eqn:einsteadofw}, this yields
\begin{align*}
    \BE_{s(w^i)} |s(R \mid w^i) - s(R)| 
    &\leq \BE_{s(x^iy^i)} |s(R \mid x^iy^i) - s(R)| \\
    &= O \left( \sqrt{\Infc[s]{R}{X^iY^i}{\{X^iY^i,\tilde{X}^i\tilde{Y}^i\}}} \right) \\
    &= O \left( \sqrt{\Infc[s]{R}{X^iY^i}{W^i}} \right),
\end{align*}
where the first inequality follows from \Cref{prop:markov2}.
Thus, it remains to prove \eqref{eqn:graphthings}.
To this end, we consider the undirected graph $\FG$ defined on the node set $g^{-1}(b^i)$, where there is an edge between two nodes iff they are contained in a common window. Here, we explicitly allow self-loops. The following is shown in~\cite{Goos}. For the sake of completeness, we also provide a proof in \Cref{secProofsOfLemmas}.

\begin{replemma}{lemshouldntbeclaim}
    There is a distribution $s'$ on the set of random walks $V_0, \dots, V_{192}$ of length 192 in $\FG$ such that $V_0$ and $V_{192}$ are independent in $s'(\cdot)$, $s'(V_k) = s(X^iY^i)$ for $k \in \{0, \dots, 192 \}$, and $s'(V_{k - 1}, V_k) = s(X^iY^i, \tilde{X}^i \tilde{Y}^i)$ for any $k \in \{1, \dots, 192\}$.
\end{replemma}

Using the random walk provided by this Lemma and setting $V = X^iY^i, U = \tilde{X}^i \tilde{Y}^i$ for conciseness of notation, we get
\begin{align*}
    \BE_{s(v)} |s(R \mid v) - s(R)|    
    &= \BE_{s'(v_0)} |s(R \mid V = v_0) - s(R)| \\
    &= \BE_{s'(v_0)} |\BE_{s'(v_{192})} s(R \mid V = v_0) - \BE_{s'(v_{192})} s(R \mid V = v_{192})| \\
    &\leq \BE_{s'(v_0)} \BE_{s'(v_{192})} |s(R \mid V = v_0) - s(R \mid V = v_{192})| \\
    &= \BE_{s'(v_0, v_{192})} |s(R \mid V = v_0) - s(R \mid V = v_{192})| \\
    &= \BE_{s'(v_0, \dots, v_{192})} |s(R \mid V = v_0) - s(R \mid V = v_{192})| \\
    &\leq \BE_{s'(v_0, \dots, v_{192})} \sum \nolimits_{k = 1}^{192} |s(R \mid V = v_{k - 1}) - s(R \mid V = v_k)| \\
    &= \sum \nolimits_{k = 1}^{192} \BE_{s'(v_{k - 1}, v_k)} |s(R \mid V = v_{k - 1}) - s(R \mid V = v_k)| \\
    &= 192 \cdot \BE_{s(v, u)} |s(R \mid V = v) - s(R \mid V = u)|,
\end{align*}
where both inequalities follow from the triangle inequality for statistical distance, and the third equality holds since $V_0$ and $V_{192}$ are independent. Note that we have 
\begin{align*}
    s(R \mid \{u, v\})
    &= s(R, V = v, U = u \mid \{u, v\}) + s(R, V = u, U = v \mid \{u, v\}) \\
    &= s(R \mid V = v, U = u) \cdot s(V = v, U = u \mid \{u, v\}) \\
    &\quad + s(R \mid V = u, U = v) \cdot s(V = u, U = v \mid \{u, v\}) \\
    &= \tfrac12 \cdot s(R \mid V = v, U = u) + \tfrac12 \cdot s(R \mid V = u, U = v) \\
    &= \tfrac12 \cdot s(R \mid V = v) + \tfrac12 \cdot s(R \mid V = u),
\end{align*}
where the last equality follows since $R - V - U$, which is a consequence of Claim \ref{clmhtoewph}.
Hence, we get
\begin{align*}
    \BE_{s(v, u)} |s(R \mid V = v) - s(R &\mid V = u)| = 2 \cdot \BE_{s(v, u)} |s(R \mid V = v) - s(R \mid \{u, v\})| \\
    &= 2 \cdot \BE_{s(u, \{u, v\})} |s(R \mid V = v, \{u, v\}) - s(R \mid \{u, v\})| \\
    &\leq 2 \cdot \sqrt{\Infc[s]{R}{V}{\{U, V\}}},
\end{align*}
where the second equality follows from $R - V - \{U, V\}$, which is again a consequence of Claim  \ref{clmhtoewph}, and the inequality follows from Pinsker's inequality (\Cref{proposition:pinsker}).
\qed
\end{proof}

\subsection{Proof of \Cref{bigprop2}}
The basic insight behind this proof is that when $(r,t,w) \in \CG_2$, then $p(A,B\mid rtw)$ is close to the uniform distribution on two bits and hence one can show that $p(A = 0, B = 0, \CG_2)$ is close to $p(A = 1, B = 1, \CG_2)$. However, the event that $A=1,B=1$ implies $\Zeros(X,Y) = T_* \cup T_i$, regardless of the particular $\binom{5}{3}$ ways of refining $T_* \cup T_i$. This means that each such possible refinement contributes $1/\binom{5}{3}$ to the probability of the event. Since all partitions considered in $\CG_2$ intersect in two elements, we can bound the total contribution by using bounds on the maximum size of 2-intersecting families of $\binom{[5]}{3}$. We now proceed with the rigorous arguments to formalize the above intuition.
    
    Let $rt_*w \in \CG_2$ be arbitrary and let $t$ be a refinement of $t_*$.
    Recall that $rt_*w \in \CG_2 \subseteq \CG$ implies $\alpha_{rt_*w} \leq \gamma$ and $\beta_{rt_*w} \leq \gamma$, which yields $p(A = 0 \mid rtw) \leq p(A = 1 \mid rtw) + \gamma$ and $p(B = 0 \mid rtw) \leq p(B = 1 \mid rtw) + \gamma$. Thus we get
    \begin{align*}
        &\quad ~ p(A = 0, \, B = 0, rtw) = p(A = 0, B = 0 \mid rtw) \cdot p(rtw) \\
        &= p(A = 0 \mid rtw) \cdot p(B = 0 \mid rtw) \cdot p(rtw) \\
        &\leq (p(A = 1 \mid rtw) + \gamma) \cdot (p(B = 1 \mid rtw) + \gamma) \cdot p(rtw) \\
        &= p(A = 1, B = 1, rtw) + p(rtw) \cdot ((p(A = 1 \mid rtw) + p(B = 1 \mid rtw)) \cdot \gamma + \gamma^2) \\
        &\leq p(A = 1, B = 1, rtw) + p(rtw) \cdot (2\gamma + \gamma^2),
    \end{align*}
    where the second and third equality follow from Claim \ref{clmuoetzwh}.
    With the above, we get
    \begin{align*}
        p(A = 0, B &= 0, \CG_2) = \sum_{rt_*w \in \CG_2} \sum_{t \text{ refining } t_*} p(A = 0, B = 0, rtw) \\
        &\leq \sum_{rt_*w \in \CG_2} \sum_{t \text{ refining } t_*} \left[ p(A = 1, B = 1, rtw) + p(rtw) \cdot (2 \gamma + \gamma^2) \right] \\
        &= p(A = 1, B = 1, \CG_2) + (2 \gamma + \gamma^2) \cdot \sum_{rt_*w \in \CG_2} \sum_{t \text{ refining } t_*} p(rtw) \\
        &\leq p(A = 1, B = 1, \CG_2) + 2 \gamma + \gamma^2 \\
        &= p(A = 1, B = 1, \CG_2) + O(\gamma),
    \end{align*}
    where the last equality holds since $\gamma$ is assumed to be small.
    It remains to show that $p(A = 1, B = 1, \CG_2) \leq \frac25 p(A = 1, B = 1)$ holds.
    To this end, we observe that every $t_*$ consistent with $t_{-i}$ and every $rw$ satisfy
        \begin{align*}
            p(t_* \mid A = 1, B = 1, rw)
            &= q(t_* \mid A = 1, B = 1, rwzt_{-i}, \CE) \\
            &= q(t_* \mid rwt_{-i}, XY \in w)
            = q(t_* \mid t_{-i})
            = \nicefrac{1}{\binom53} = \nicefrac{1}{10},
        \end{align*}
    where the third equality follows from the independence of $XRY$, $T$ and $W$ in $q(\cdot)$ and the fourth equality holds since $q(T)$ is uniform.
    Thus, we have
    \begin{align*}
        p(A = 1, B = 1, \CG_2) &= \sum_{rt_*w \in \CG_2} p(t_* \mid A = 1, B = 1, rw) \cdot p(A = 1, B = 1, rw) \\
        &= \frac{1}{10} \cdot \sum_{rt_*w \in \CG_2} p(A = 1, B = 1, rw) \\
        &= \frac{1}{10} \cdot \sum_{rw} |\{t_* : rt_*w \in \CG_2\}| \cdot p(A = 1, B = 1, rw).
    \end{align*}
    Finally, we claim that $|\{t_* : rt_*w \in \CG_2\}| \le 4$ holds for every $rw$, which yields the desired inequality $p(A = 1, B = 1, \CG_2) \le \frac25 p(A = 1, B = 1)$.
    To see this, fix any $rw$ and recall that since $t_{-i}$ is already fixed, there are five nodes left to choose $t_*$ from.
    By the definition of $\CG_2$, all $t_*^1, t_*^2$ with $rt_*^1w,rt_*^2w \in \CG$ satisfy $|t_*^1 \cap t_*^2| \geq 2$.
    Therefore, by Erdos-Ko-Rado type bounds~\cite{W84} on the size of intersecting families\footnote{Note that the bound in~\cite{W84} says that for any $n \ge 6$, the size of any $2$-intersecting family of $\binom{[n]}{3}$ is at most $n-2$. The bound here follows from the case of $n=6$. One can also consult ~\cite{sinha2018lower} for a self-contained proof of this fact.}, we obtain $|\{t_* : rt_*w \in \CG_2\}| \leq 4$.
    
\subsection{Proof of \Cref{bigprop3}}
    In what follows, all $(r,t,w)$ considered will lie in $\CG$ and the reader should tacitly assume this in all the statements below unless specified otherwise. The following two propositions follow from the fact that conditioned on $(r,t,w) \in \CG$, the distribution  $p(A,B \mid r,t,w)$ is close to uniform. We defer the proofs to the end of this section. 
    \begin{lemma}
        \label{lemasfhuafg}
        For $\gamma > 0$ small enough and any $a,b,a',b' \in \bits$, we have $\frac{p(A = a, B = b, rtw)}{p(A = a', B = b', rtw)} \in \left[\tfrac12,2\right]$.
    \end{lemma}
    \begin{lemma}
        \label{lemasg78gfs}
        For $\gamma > 0$ small enough, $a,b,a',b' \in \bits$, and $t,t'$ refining $t_{-i}$ we have $\frac{p(A = a, B = b, rtw)}{p(A = a', B = b', rt'w)} \in \left[\tfrac14,4\right]$.
    \end{lemma}

    Consider any $rt_*w \in \CG_1$. Note that there exists some $t'_* \neq t_*$ with $rt'_*w \in \CG$ (otherwise, $rt_*w \in \CG_2$). Note also that $rt'_*w \in \CG_1$ (otherwise, $rt_*w \notin \CG$).
     Among those $t'_*$, by the definition of $\CG_1$, we can choose one that further satisfies $|t_* \cap t'_*| = 1$. Note that we can make this choice in a symmetric way, meaning that if $t'_*$ is chosen for $t_*$, then $t_*$ is chosen for $t'_*$. Let $\sigma(t_*)$ denote the unique refinement of $t_*$ whose unpaired node is the unique node in $t_* \cap t'_*$.   
    Let us denote by $\CG_1^\sigma$ the event that $R= r, T = \sigma(t_*), W = w$ for some $rt_*w \in \CG_1$. Then, we have
    \begin{align*}
        p(A=0,B=0,\CG_1)
        & = \sum_{rt_*w \in \CG_1} \sum_{t \text{ refining } t_*} p(A = 0, B = 0, rtw) \\
        & \leq \sum_{rt_*w \in \CG_1} \sum_{t \text{ refining } t_*} 4 \cdot p(A = 0, B = 0, rw, T = \sigma(t_*)) \\
        & = \sum_{rt_*w \in \CG_1} 12 \cdot p(A = 0, B = 0, rw, T = \sigma(t_*)) \\
        & = 12 \cdot p(A = 0, B = 0, \CG_1^\sigma),
    \end{align*}
    where the inequality follows from \Cref{lemasg78gfs}.
    It remains to show that $p(A = 0, B = 0, \CG_1^\sigma) = O(\eps)$.
    For this,  let us define $\CF \supset \CE$ to be the event that $X^iY^i \in \ps{W^i}{ab}$ and $X^0Y^0 \in \ps{W^0}{cd}$ for some $a,b,c,d,\in \bits$, and let $ABCD$ denote the corresponding random variables. Since $\ps{W^0}{11} = W^0$, it follows that $\CE$ is equivalent to the event that there exist $a,b \in \bits$ such that $X^iY^i \in \ps{W^i}{ab}$ and $X^0Y^0 \in \ps{W^0}{11}$, i.e., $A,B \in \bits$ and $C=D=1$. Thus, we have
    \begin{equation}
    \label{eqasdag8s8as}
        \begin{split}
        &\quad ~ p(A=0,B=0,\CG^\sigma_1) \\
        &= q(A = 0, B = 0, \CG^\sigma_1 \mid zt_{-i}, \CE) \\
        &= q(\CE \mid zt_{-i}, \CF)^{-1} \cdot q(A = 0, B = 0, \CG^\sigma_1, \CE \mid zt_{-i}, \CF) \\
        &= q(\CE \mid zt_{-i}, \CF)^{-1} \cdot q(A = 0, B = 0, C = 1, D = 1, \CG^\sigma_1 \mid zt_{-i}, \CF) \\
        &= q(\CE \mid zt_{-i}, \CF)^{-1} \cdot \sum_{rtw \in \CG^\sigma_1} q(A = 0, B = 0, C = 1, D = 1, rtw \mid zt_{-i}, \CF).
        \end{split}
    \end{equation}
    Below, we will show that if $rtw \in \CG^\sigma_1$, then
    \begin{equation}
        \label{eqn:zeroevent}
        \begin{split}
        &\quad ~ q(A = 0, B = 0, C = 1, D = 1 \mid rtw, \CF) \\
        &\leq 16 \cdot q(A = 0, B = 1, C = 1, D = 0 \mid rtw, \CF).
        \end{split}
    \end{equation}
    With this, for every $rtw \in \CG^\sigma_1$ we get
    \begin{align*}
        q(A = 0, \, & B = 0, C = 1, D = 1, rtw \mid zt_{-i}, \CF) \\
        & = q(A = 0, B = 0, C = 1, D = 1 \mid rtw, \CF) \cdot q(rtw \mid zt_{-i}, \CF) \\
        & \le 16 \cdot q(A = 0, B = 1, C = 1, D = 0 \mid rtw, \CF) \cdot q(rtw \mid zt_{-i}, \CF) \\
        & = 16 \cdot q(A = 0, B = 1, C = 1, D = 0, rtw \mid zt_{-i}, \CF),
    \end{align*}
    and hence, by~\eqref{eqasdag8s8as} we obtain
    \begin{align*}
        &\quad ~ p(A=0,B=0,\CG^\sigma_1) \\
        &\le q(\CE \mid zt_{-i}, \CF)^{-1} \cdot 16 \cdot q(A = 0, B = 1, C = 1, D = 0, \CG^\sigma_1 \mid zt_{-i}, \CF) \\
        &\leq q(\CE \mid zt_{-i}, \CF)^{-1} \cdot 16 \cdot q(|\Zeros(X, Y)| = 1, \CG^\sigma_1 \mid zt_{-i}, \CF) \\
        &\leq q(\CE \mid zt_{-i}, \CF)^{-1} \cdot 16 \cdot q(|\Zeros(X, Y)| = 1 \mid zt_{-i}, \CF).
    \end{align*}
The second line follows since the event that $A=0,B=1,C=1,D=0$ is a subset of the event that $|\zeros(X,Y)|=1$. To see this, recall that we have $\Zeros(W) = z = T_u \cup T_0 \cup T_i$ by Claim \ref{clmztuior}. Since replacing $W^i$ and $W^0$ by $\ps{W^i}{01}$ and $\ps{W^0}{10}$ flips the value of the gadget along the paths corresponding to $T_i$ and $T_0$ (by Claim \ref{clmwindowa}), this shows $\Zeros(X, Y) = T_u$.
    
Notice that we have related the contribution coming from $\CG_1^\sigma$ to the event $|\zeros(X,Y)|=1$. Since whenever $|\zeros(x,y)|=1$, the corresponding slack matrix entry is $S_{xy} = \eps$, this gives us the following claim which then finishes the proof. 
    \begin{repclm}{jlbopf}
        $\frac{q(|\Zeros(X, Y)| = 1 \ \mid \ zt_{-i}, \CF)}{q(\CE \mid zt_{-i}, \CF)} = O(\eps)$.
    \end{repclm}
    It remains to prove \eqref{eqn:zeroevent}. Let $rtw \in \CG_1^\sigma$. From the definition of $\CG_1^\sigma$, it follows that there exists a partition $t'$ such that $t'_* \cap t_* = t_u = t'_u$ and $rt'w \in \CG_1^\sigma$. Observe that in this case, $t'_0 = t_i$ and $t'_i = t_0$, so the pairing of all the terminals is the same apart from the indices $i$ and $0$ being swapped. Moreover, we observe the following:
    \begin{repclm}{thewiufd}
        $AC - R - BD$ in $q( \cdot \mid tw, \CF)$.
    \end{repclm}
    
    From Claim \ref{thewiufd} and \Cref{lemasg78gfs}, we get
    \begin{align*}
        &\quad ~ \frac{q(A = 0, B = 0, C = 1, D = 1 \mid rtw, \CF)}
                {q(A = 0, B = 1, C = 1, D = 0 \mid rtw, \CF)} \\
        &= \frac{q(A = 0, C = 1 \mid rtw, \CF)}{q(A = 0, C = 1 \mid rtw, \CF)}
        \cdot \frac{q(B = 0, D = 1 \mid rtw, \CF)}{q(B = 1, D = 0 \mid rtw, \CF)} \\
        &= \frac{q(A = 1, C = 1 \mid rtw, \CF)}{q(A = 1, C = 1 \mid rtw, \CF)}
        \cdot \frac{q(B = 0, D = 1 \mid rtw, \CF)}{q(B = 1, D = 0 \mid rtw, \CF)} \\
        &= \frac{q(A = 1, B = 0, C = 1, D = 1 \mid rtw, \CF)}
                {q(A = 1, B = 1, C = 1, D = 0 \mid rtw, \CF)} \\
        &= \frac{q(A = 1, B = 0, C = 1, D = 1 \mid rtw, \CF)}
                {q(A = 1, B = 0, C = 1, D = 1 \mid rt'w, \CF)} \\
        &\quad \cdot \frac{q(A = 1, B = 1, C = 1, D = 1 \mid rt'w, \CF)}
                {q(A = 1, B = 1, C = 1, D = 1 \mid rtw, \CF)} \\
        &= \frac{p(A = 1, B = 0 \mid rtw)}{p(A = 1, B = 0 \mid rt'w)} 
        \cdot \frac{p(A = 1, B = 1 \mid rt'w)}{p(A = 1, B = 1 \mid rtw)} \\
        &= \frac{p(A = 1, B = 0, rtw)}{p(A = 1, B = 1, rtw)}
        \cdot \frac{p(A = 1, B = 1, rt'w)}{p(A = 1, B = 0, rt'w)} \\
        & \leq 4 \cdot 4 = 16,
    \end{align*}
    where the fourth equality uses $q(A = B = C = D = 1 \mid rtw, \CF) = q(XY \in w \mid rtw, \CF) = q(XY \in w \mid rt'w, \CF) = q(A = B = C = D = 1 \mid rt'w, \CF)$.
    Thus, it remains to prove \Cref{lemasfhuafg} and \Cref{lemasg78gfs}.

\begin{proof}[Proof of \Cref{lemasfhuafg}]
    Let us consider any $rt_*w \in \CG$ and any refinement $t$ of $t_*$.
    Recall that $A - R - B$ in $p(\cdot \mid tw)$ due to Claim \ref{clmuoetzwh} and thus, for any $a, b \in \bits$, we have that
    \begin{align*}
        &\quad ~ 4 \cdot p(A = a, B = b \mid rtw) \\
        &\leq (p(A = \comp{a} \mid rtw) + \gamma) \cdot (p(B = \comp{b} \mid rtw) + \gamma) + p(A = a \mid rtw) \cdot p(B = b \mid rtw) \\
        &+ p(A = a \mid rtw) \cdot (p(B = \comp{b} \mid rtw) + \gamma) + (p(A = \comp{a} \mid rtw) + \gamma) \cdot p(B = b \mid rtw)
         \\
        &\leq p(A \in \{a, \comp{a} \}, B \in \{b, \comp{b} \} \mid rtw) + 2  \gamma + \gamma^2 \\
        &= 1 + 2 \gamma + \gamma^2.
    \end{align*}
    This also implies
    \[
        p(A = a, B = b \mid rtw) = 1 - \sum_{(a', b') \neq (a, b)} p(A = a', B = b' \mid rtw) \geq 1 - \frac34 \cdot (1 + 2 \gamma + \gamma^2).
    \]
    Thus, choosing $\gamma > 0$ such that $ 0 < \frac{1 + 2 \gamma + \gamma^2}{1 - 6 \gamma - 3 \gamma^2} \le 2$, we obtain
    \[
        \frac{p(A = a, B = b, rtw)}{p(A = a', B = b', rtw)}
        = \frac{p(A = a, B = b \mid rtw)}{p(A = a', B = b' \mid rtw)}
        \in \left[\tfrac12,2\right]. \quad \qed
    \]
\end{proof}

\begin{proof}[Proof of \Cref{lemasg78gfs}]
    We have
    \begin{align*}
        \frac{p(A = a, B = b, rtw)}{p(A = a', B = b', rt'w)}
        &= \frac{p(A = a, B = b, rtw)}{p(A = a', B = b', rt'w)} \cdot \frac{q(XY \in w, rw) \cdot q(t')}{q(XY \in w, rw) \cdot q(t)} \\
        & = \frac{p(A = a, B = b, rtw)}{p(A = a', B = b', rt'w)}
            \cdot \frac{p(A = 1, B = 1, rt'w)}{p(A = 1, B = 1, rtw)} \\
        & = \frac{p(A = a, B = b, rtw)}{p(A = 1, B = 1, rtw)}
            \cdot \frac{p(A = 1, B = 1, rt'w)}{p(A = a', B = b', rt'w)} \\
            &\in \left[ \tfrac14, 4 \right],
    \end{align*}
    where the latter follows from \Cref{lemasfhuafg}.
    \qed
\end{proof}

\subsubsection{Acknowledgements} We would like to thank Mel Zürcher for initial discussions on this topic.

\bibliographystyle{splncs04}
\bibliography{references}

\appendix
\section{Lower Bounds for MIP formulations for Knapsack}
\label{secKnapsack}

In this section, we provide a proof of \Cref{thmMain} for the case of the knapsack problem, assuming that it holds already for the stable set problem.
To this end, we will give a polyhedral reduction between the two problems making use of an alternative definition of MIP formulations, which is more geometric and coincides with the notion used in \cite{Cevallos}.

\begin{lemma}
    \label{lemMIPdef}
    Let $\CV$ be a finite set and $\CF$ be a family of subsets of $\CV$.
    If $(Q,\sigma)$ is a MIP formulation of $(\CV,\CF)$, then there exists an affine subspace $L$ such that $P(\CV,\CF)$ is an affine projection of $\conv (\Gamma(Q \cap L,\sigma))$.
    Conversely, if $Q \subseteq \R^d$ is a polyhedron and $\sigma : \R^d \to \R^k$ is an affine map such that $P(\CV,\CF)$ is an affine projection of $\conv (\Gamma(Q,\sigma))$, then $(Q,\sigma)$ is a MIP formulation of $(\CV,\CF)$.
\end{lemma}
\begin{proof}
    To prove the first claim, let $(Q,\sigma)$ be a MIP formulation of $(\CV,\CF)$.
    Let $P(\CV,\CF) = \{y \in \R^n : Wy \le h\}$, where $W \in \R^{p \times n}$, $h \in \R^p$.
    Moreover, let $\gamma : \R^n \to \R^p$ be the affine map that is given by $\gamma(y) = h - Wy$.
    Note that $\gamma$ is injective, and hence we may define the affine map $\nu : \gamma(\R^n) \to \R^n$ as the inverse of $\gamma$.

    For every $S \in \CF$, let $x_S$ denote the corresponding vector in $\Gamma = \Gamma(Q,\sigma)$.
    For $i \in [p]$, think of the $i$-th row $W_{i,*}$ of $W$ as weights on $\CV$.
    By the definition of a MIP formulation, there is an affine map $c_i : \R^d \to \R$ with $\max\{c_i(x) : x \in \Gamma\} \le h_i$ such that $W_{i,*} \chi(S) = c_i(x_S)$ holds for all $S \in \CF$.
    Here, $h_i$ denotes the $i$-th entry of $h$.
    Thus, defining the affine map $\tau : \R^d \to \R^p$ via $\tau_i(x) = h_i - c_i(x)$, we see that $\tau(x) \ge \mathbf{0}$ for all $x \in \Gamma$ and $\tau(x_S) = h - W\chi(S) = \gamma(\chi(S))$ for all $S \in \CF$.

    Letting $L = \{x \in \R^d: \tau(x) \in \gamma(\R^n)\}$ and $\pi := \nu \circ \tau: L \to \R^n$, it remains to show that $P(\CV,\CF) = \conv(\pi(\Gamma \cap L))$ holds.
    To this end, we denote the right-hand side by $P'$.
    Consider any set $S \in \CF$.
    Since $\tau(x_S) = \gamma(\chi(S))$ and $x_S \in \Gamma$, we have $\pi(x_S) \in P'$.
    Moreover, we have $\pi(x_S) = \nu(\tau(x_S)) = \nu(\gamma(\chi(S))) = \chi(S)$.
    This shows $P(\CV,\CF) \subseteq P'$.
    To see that $P' \subseteq P(\CV,\CF)$ holds, consider any $x \in \Gamma \cap L$.
    Since $x \in L$, there is some $y \in \R^n$ with $\gamma(y) = \tau(x)$.
    Note that we have $y \in P(\CV,\CF)$ since $\tau(x) \ge \mathbf{0}$.
    We conclude $\pi(x) = \nu(\tau(x)) = \nu(\gamma(y)) = y \in P(\CV,\CF)$.

    To prove the second claim, suppose that $P(\CV,\CF) = \pi(\conv(\Gamma))$ for some affine map $\pi : \R^d \to \R^n$ and $\Gamma = \Gamma(Q,\sigma)$, where $Q \subseteq \R^d$ and $\sigma : \R^d \to \R^k$.
    Clearly, for every $S \in \CF$, the point $\chi(S)$ is contained in $\conv(\pi(\Gamma))$.
    Since $\chi(S)$ is a vertex of $\conv(\pi(\Gamma))$, we even have $\chi(S) \in \pi(\Gamma)$.
    Thus, for every $S \in \CF$ there is a point $x_S \in \Gamma$ with $\pi(x_S) = \chi(S)$.
    For every node weights $w : \CV \to \R$ we also consider $w$ as a vector in $\R^n$ and define the corresponding affine map $c_w : \R^d \to \R$ via $c_w(x) = w^\intercal \pi(x)$.
    Note that we have $c_w(x_S) = w^\intercal \pi(x_S) = w^\intercal \chi(S) = w(S)$ for every $S \in \CF$ as well as
    \begin{align*}
        \max\{c_w(x) : x \in \Gamma\}
        = \max\{w^\intercal y : y \in \pi(\Gamma)\}
        & = \max\{w^\intercal y : y \in \pi(\conv(\Gamma))\} \\
        & = \max\{w^\intercal y : y \in P(\CV,\CF)\} \\
        & = \max\{w(S) : S \in \CF\}.
    \end{align*}
    This shows that $(Q,\sigma)$ is indeed a MIP formulation of $(\CV,\CF)$. \qed
\end{proof}

The idea for proving \Cref{thmMain} for the knapsack problem is as follows.
We consider any graph $G$ as in the statement of \Cref{thmMain} and its stable set polytope $P$.
Using a polyhedral reduction from Pokutta \& Van Vyve~\cite{pokutta2013note}, we construct a knapsack instance such that a face $F$ of the corresponding knapsack polytope can be affinely projected onto $P$.
Now, let $(Q,\sigma)$ be any MIP formulation for this knapsack problem, and suppose it has size $m$ and $k$ integer variables.
By \Cref{lemMIPdef}, there exists an affine subspace $L$ such that the knapsack polytope is an affine projection of $X := \conv (\Gamma(Q \cap L,\sigma))$.
Since $F$ is a face of the knapsack polytope, this means that there is a face $X'$ of $X$ that affinely projects onto $F$, and so $X'$ also affinely projects onto $P$.
Again by \Cref{lemMIPdef}, this yields a MIP formulation for the stable set problem over $G$ with $k$ integer variables and size at most $m$.
Since we assume that \Cref{thmMain} holds for the stable set problem, this will yield the desired bounds on $m$ and $k$.

Let us now formalize the above idea.
We assume that \Cref{thmMain} holds for the stable set problem, that is, there is a constant $c > 0$ and a family of graphs such that every MIP formulation for the stable set problem of an $n$-node graph in this family requires $\Omega(\nicefrac{n}{\log^2 n})$ integer variables, unless its size is at least $2^{cn/\log n}$.
In the proof of \Cref{thmMain} we will give a family of such graphs for which their number of edges is linear in the number of nodes.
Let $G = (V,E)$ be a graph in this family with $|V| = n$.

Consider the following knapsack instance on $N = n + |E| = \Theta(n)$ items that we identify with the elements in $V \cup E$.
We label the edges in $E$ by $e_1,\dots,e_{|E|}$ and define the sizes of the items $a : V \cup E \to \R_{\ge 0}$ via
\[
    a(e_j) = 4^j \text{ for all } j \in [|E|], \qquad a(v) = \sum \nolimits_{e : v \in e} a(e) \text{ for all } v \in V,
\]
and the capacity is given by
\[
    B = \sum_{e \in E} a(e).
\]
Thus, the feasible subsets of the knapsack problem correspond to the points in
\[
    X = \left\{x \in \{0,1\}^{V \cup E} : \sum \nolimits_{v \in V} a(v)x_v + \sum \nolimits_{e \in E} a(e)x_e \le B \right\}.
\]
Suppose now that $(Q,\sigma)$ is a MIP formulation for this knapsack instance, where $Q \subseteq \R^d$ is a polyhedron with $m$ facets and $\sigma : \R^d \to \R^k$.
By \Cref{lemMIPdef}, there exists an affine subspace $L \subseteq \R^d$ and an affine map $\pi : \R^d \to \R^{V \cup E}$ such that $\pi(\conv (\Gamma'))$ is equal to the knapsack polytope $P = \conv(X)$, where $\Gamma' = Q \cap L \cap \sigma^{-1}(\Z^k)$.
Consider the hyperplane $H = \{ x \in \R^{V \cup E} : \sum \nolimits_{v \in V} a(v)x_v + \sum \nolimits_{e \in E} a(e)x_e = B\}$ and the set $F = P \cap H$.
Since $F$ is a face of $P$ and $P = \pi(\conv(\Gamma'))$, we see that $\conv(\Gamma') \cap \pi^{-1}(H)$ is a face of $\conv(\Gamma')$.
In particular, we obtain $\conv(\Gamma') \cap \pi^{-1}(H) = \conv(\Gamma' \cap \pi^{-1}(H))$.
Thus, we obtain
\begin{multline*}
    F
    = P \cap H
    = \pi(\conv(\Gamma')) \cap H 
    = \pi(\conv(\Gamma') \cap \pi^{-1}(H)) \\
    = \pi(\conv(\Gamma' \cap \pi^{-1}(H)))
    = \pi(\conv(Q' \cap \sigma^{-1}(\Z^k))),
\end{multline*}
where $Q' = Q \cap L \cap \pi^{-1}(H)$.
Note that $Q'$ is a polyhedron with at most $m$ facets.
Moreover, we claim that the stable set polytope of $G$ is a linear projection of $F$.
By \Cref{lemMIPdef}, this implies that $(Q',\sigma)$ is a MIP formulation for the stable set problem over $G$ and hence $k = \Omega(\nicefrac{n}{\log^2 n}) = \Omega(\nicefrac{N}{\log^2 N})$, unless $m \ge 2^{cn/\log n} = 2^{\Omega(N/\log N)}$, and we are done.

To see that the stable set polytope of $G$ is indeed a linear projection of $F$, consider the projection $\tau : \R^{V \cup E} \to \R^V$ onto the variables indexed by $V$.
First, let $S \subseteq V$ be any stable set of $G$ and denote by $E' \subseteq E$ the set of edges that are disjoint from $S$.
Let $x \in \{0,1\}^{V \cup E}$ denote the characteristic vector of $S \cup E'$ and note that
\begin{multline*}
    \sum_{v \in V} a(v)x_v + \sum_{e \in E} a(e)x_e
    = \sum_{v \in S} a(v) + \sum_{e \in E'} a(e) \\
    = \sum_{v \in S} \sum_{e : v \in e} a(e) + \sum_{e \in E'} a(e)
    = \sum_{e \in E} a(e) = B,
\end{multline*}
where the second equality uses the fact that $S$ is a stable set.
Thus, $x \in F$ and hence $\chi(S) = \tau(x) \in \tau(F)$.
This shows that the stable set polytope of $G$ is contained in $\tau(F)$.
For the reverse inclusion, consider any $x \in F \cap \{0,1\}^{V \cup E}$ and note that it suffices to show that $\tau(x)$ is the characteristic vector of a stable set in $G$.
To this end, let $S = \{v \in V : x_v = 1\}$.
Since $\tau(x) = \chi(S)$, we have to show that $S$ is a stable set in $G$.
For $j \in [|E|]$ let $e_j = \{v_j,w_j\}$ and notice that
\[
    \sum_{j=1}^{|E|} 4^j
    = B
    = \sum_{v \in V} a(v)x_v + \sum_{e \in E} a(e)x_e
    = \sum_{j=1}^{|E|} 4^j (x_{e_j} + x_{v_j} + x_{w_j}).
\]
In order for this equality to hold, for every $j \in [|E|]$ we must have $x_{e_j} + x_{v_j} + x_{w_j} = 1$ and hence $x_{v_j} + x_{w_j} \le 1$.
This means that $S$ contains at most one node from every edge of $E$, and hence $S$ is indeed a stable set.
\section{Proof of \Cref{lemshouldntbeclaim}}
\label{secProofsOfLemmas}

    For $b \in \bits$, consider the undirected graph $\FG^b = (\CV^b, \CE^b)$ defined on the node set $g^{-1}(b)$, where there is an edge between two nodes iff they are contained in a common window. We explicitly allow self-loops in $\FG^b$, and denote the version without self-loops by $\tilde{\FG}^b$. Let $b_1, \dots, b_\ell$ denote the entries of $b^i$ corresponding to the edges of the path joining the terminals in $t_i$. Then, we can write $\FG = (\CV, \CE)$ as a tensor product of graphs via $\FG = \FG^{b_1} \otimes \dots \otimes \FG^{b_\ell}$, meaning that $\CV = \CV^{b_1} \times \dots \times \CV^{b_\ell}$ and $\{u_1 \times \dots \times u_\ell, v_1 \times \dots \times v_\ell \} \in \CE$ iff $\{u_1, v_1\} \in \CE^{b_1}, \dots, \{u_\ell, v_\ell\} \in \CE^{b_\ell}$.

    We will later prove that for any $b \in \bits$, there exists a distribution $s'_b$ on the set of random walks $V_0, \dots, V_{192}$ of length 192 in $\FG^b$ such that $V_0$ and $V_{192}$ are independent in $s'_b(\cdot)$, $s'_b(V_k)$ is uniform in $\CV^b$ for all $k \in \{0, \dots, 192\}$, $s'_b(V_k = V_{k - 1} \mid v_{k - 1}) = \frac12$ and $s'_b(V_k \mid V_k \neq V_{k - 1}, v_{k - 1})$ is uniform over the neighbors of $v_{k - 1}$ (except for $v_{k - 1}$ itself) for all $k \in \{1, \dots, 192\}$ and $v_{k - 1} \in \CV^b$.
    In particular, this implies that for any two neighbors $v_{k - 1}, v_k \in \CV^b$ with $v_{k - 1} \neq v_k$, we have $s'_b(v_{k - 1}) = \frac{1}{32}$ and $s'_b(v_k \mid V_k \neq V_{k - 1}, v_{k - 1}) = \frac16$, since each node in $\FG^b$ has degree 6.

    First, let us show that \Cref{lemshouldntbeclaim} is implied by the existence of such distributions of random walks. To this end, let us independently sample $\ell$ random walks $(V_0^1, \dots, V_{192}^1), \dots, (V_0^\ell, \dots, V_{192}^\ell)$ on $\FG^{b_1}, \dots, \FG^{b_\ell}$ according to the distributions $s'_{b_1}, \dots, s'_{b_\ell}$ and let $s' = s'_{b_1} \cdot \dots \cdot s'_{b_\ell}$ be the product distribution. For $k \in \{0, \dots, 192\}$, we set $V_k = (V_k^1, \dots, V_k^\ell)$ to create a random walk $(V_0, \dots, V_{192})$ in $\FG$. For each $k \in \{1, \dots, n\}$, $V_0^k$ and $V_{192}^k$ are independent in $s'_{b_k}(\cdot)$. Thus, $V_0$ and $V_{192}$ are independent in the product distribution $s'(\cdot)$.
    We observe the following.
    \begin{repclm}{clmhkfdsub}
        In $s(\cdot)$, $X^iY^i$ is uniformly distributed in $g^{-1}(b^i)$.
    \end{repclm}
    \begin{repclm}{clmhkfdsuc}
        In $s(\cdot)$, $W^i$ is a uniformly random $\comp{b^i}$-window such that $\ps{W^i}{a 0}$ contains $X^iY^i$.
    \end{repclm}
    
    To see that $s'(V_{k - 1}, V_k) = s(X^iY^i, \tilde{X}^i \tilde{Y}^i)$ holds true, let $v_{k - 1}^1, \dots, v_{k - 1}^\ell$ and $v_k^1, \dots, v_k^\ell$ be fixed, and set $x^iy^i = v_{k - 1}$ and $\tilde{x}^i \tilde{y}^i = v_k$. Let $\FW^i$ denote the set of windows $w^i$ that contain both $v_{k - 1}$ and $v_k$. Note that $|\FW^i| = 6^{|\{j \in \{1, \dots, \ell\}: v_{k - 1}^j = v_k^j|}$ since each $xy \in \CX \CY$ is contained in 6 different windows.   
    Then, on one hand, we have
    \begin{align*}
        s(x^iy^i, \tilde{x}^i \tilde{y}^i) &= \sum_{w^i \in \FW^i} s(x^iy^i, \tilde{x}^i \tilde{y}^i, w^i) \\
        &= \sum_{w^i \in \FW^i} s(x^iy^i) \cdot s(w^i \mid x^iy^i) \cdot s(\tilde{x}^i\tilde{y}^i \mid w^i, x^iy^i) \\
        &= \sum_{w^i \in \FW^i} s(x^iy^i) \cdot s(w^i \mid x^iy^i) \cdot s(\tilde{x}^i\tilde{y}^i \mid w^i) \\
        &= \sum_{w^i \in \FW^i} 32^{-\ell} \cdot 6^{-\ell} \cdot 2^{-\ell} 
        = |\FW^i| \cdot 384^{-\ell},
    \end{align*}
    where the third equality follows from the independence of $X^iY^i$ and $\tilde{X}^i\tilde{Y}^i$ in $s(\cdot \mid w^i)$ and the fourth equality holds since the involved distributions are uniform (by Claim \ref{clmhkfdsub} and Claim \ref{clmhkfdsuc}). On the other hand, we have
    \begin{align*}
        s'(v_{k - 1}&, v_k)
        = \prod_{j \in \{1, \dots, \ell\}} \left[ s'_{b_j}(v_{k - 1}^j, v_k^j, V_{k - 1}^j = V_k^j) + s'_{b_j}(v_{k - 1}^j, v_k^j, V_{k - 1}^j \neq V_k^j) \right] \\
        &= \prod_{j: v_{k - 1}^j = v_k^j} s'_{b_j}(v_{k - 1}^j, v_k^j, V_{k - 1}^j = V_k^j)
        \cdot \prod_{j: v_{k - 1}^j \neq v_k^j} s'_{b_j}(v_{k - 1}^j, v_k^j, V_{k - 1}^j \neq V_k^j).
    \end{align*}
    We further calculate
    \begin{align*}
        s'_{b_j}(v_{k - 1}^j, v_k^j, V_{k - 1}^j = V_k ^j)
        &= s'_{b_j}(v_{k - 1}^j, V_{k - 1}^j = V_k ^j) \\
        &= s'_{b_j}(v_{k - 1}^j) \cdot s'_{b_j}(V_{k - 1}^j = V_k^j \mid v_{k - 1}^j) 
        = 32^{-1} \cdot 2^{-1}
    \end{align*}
    for $v_{k - 1}^j = v_k^j$ and
    \begin{align*}
        &\quad ~ s'_{b_j}(v_{k - 1}^j, v_k^j, V_{k - 1}^j \neq V_k^j) \\
        & = s'_{b_j}(v_k^j \mid v_{k - 1}^j, V_{k - 1}^j \neq V_k^j) \cdot s_{b_j}'(V_{k - 1}^j \neq V_k^j \mid v_{k - 1}^j) \cdot s_{b_j}'(v_{k - 1}^j) \\
        & = 6^{-1} \cdot 2^{-1} \cdot 32^{-1}
    \end{align*}
    for $v_{k - 1}^j \neq v_k^j$, which shows that we have
    \begin{align*}
        s'(v_{k - 1}, v_k)
        &= \prod_{j: v_{k - 1}^j = v_k^j} 32^{-1} \cdot 2^{-1}
        \cdot \prod_{j: v_{k - 1}^j \neq v_k^j} 6^{-1} \cdot 2^{-1} \cdot 32^{-1} \\
        &= 6^{|\{j: v_{k - 1}^j = v_k^j\}|} \cdot 384^{-\ell}
        = |\FW^i| \cdot 384^{-\ell}
        = s(x^iy^i, \tilde{x}^i\tilde{y}^i).
    \end{align*}
    
    It remains to prove that the distributions $s'_0$ and $s'_1$ with the stated properties exist. To this end, let us fix $b \in \bits$.
    
    With Figure~\ref{figGraph} we see that $\tilde{\FG}^b$ is connected and each node of $\tilde{\FG}^b$ has degree $6$.
    Thus, there exists a Eulerian cycle $(u_0, \dots, u_{191})$ of length $6 \cdot 32 = 192$ in $\tilde{\FG}^b$ that visits each node exactly six times. We differentiate between all the nodes $u_0, \dots, u_{191}$, even if they originated from the same node in $\tilde{\FG}^b$.

    \begin{figure}
        \centering
        \includegraphics[width=4cm]{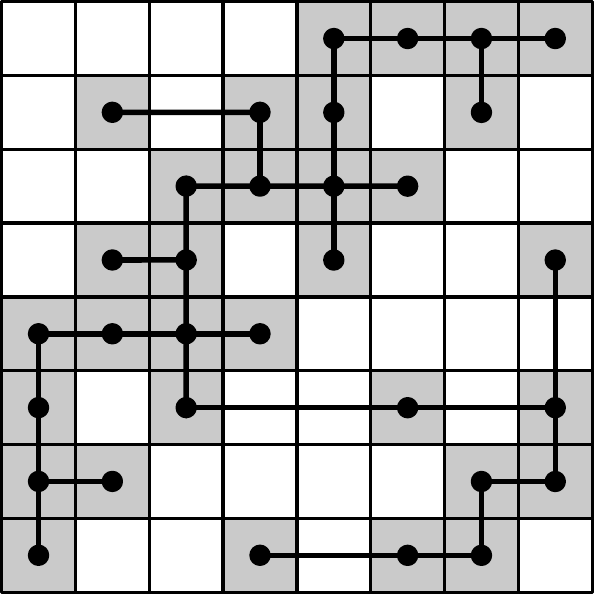}
        \caption{Illustration of a spanning subgraph of $\tilde{\FG}^1$, which shows that $\tilde{\FG}^1$ is connected. Note that every column and row has exactly $4$ gray entries, which means that every node in $\tilde{\FG}^1$ has degree $6$.
        Moreover, for every column, there is a column with inverted entries, which shows that $\tilde{\FG}^0$ and $\tilde{\FG}^1$ are isomorphic.}
        \label{figGraph}
    \end{figure}
    
    In order to choose $V_0$ and $V_{192}$, we first independently and uniformly sample two nodes $u$ and $u'$ in $\{u_0, \dots, u_{191}\}$. Without loss of generality, let us assume $u = u_0$ and let $u' = u_k$. Next, we choose to walk along the Eulerian cycle with either ascending indices or descending indices, both with probability $\tfrac12$. This results in a random walk of length $k$ or $192 - k$. Finally, depending on the previous choice, we add $192 - k$ or $k$ self-loop edges, uniformly distributed along the random walk. It is easy to see that the resulting distribution of random walks has the desired properties.
\section{Proofs of Claims}
\label{secClaims}

\repeatclm{clmwindowa}
\begin{proof}
    All calculations in this proof are done modulo 2.
    Due to symmetry reasons, it suffices to consider a horizontal window $w = \{(x,y), \, (x,y')\}$.
    We first notice that we must have $y_2 \ne y_2'$ or $y_3 \ne y_3'$ since otherwise $g(x,y) = g(x,y')$ implies $y_1 = y_1'$ and hence $y = y'$.
    This means that $\tilde{y} \ne \tilde{y}'$.
    Moreover, we have
    \begin{align*}
        g(x, \tilde{y}) - g(x, \tilde{y}')
        &= \tilde{y}_1 - \tilde{y}_1' + x_2 (\tilde{y}_2 - \tilde{y}_2') + x_3 (\tilde{y}_3 - \tilde{y}_3') \\
        &= y_1 - y_1' + x_2 (y_2 - y_2') + x_3 (y_3' - y_3) 
        = g(x, y) - g(x, y')
    \end{align*}
    and, by similar calculations, $g(\tilde{x}, y) - g(\tilde{x}, y') = g(x, y) - g(x, y')$ as well as $g(\tilde{x}, \tilde{y}) - g(\tilde{x}, \tilde{y}') = g(\tilde{x}, y) - g(\tilde{x}, y')$. Since $g$ is constant on $w$, this implies that $g$ is constant on $\ps{w}{ab}$ for any $a, b \in \bits$. Furthermore, direct calculations show $g(x, \tilde{y}) = g(\tilde{x}, y) = g(x, y) + 1$ and $g(\tilde{x}, \tilde{y}) = g(x, y) + 1 + (1 + y_2 + y_2')(1 + y_3 + y_3')$. Recall that we have $y_2 \neq y_2'$ or $y_3 \neq y_3'$, therefore the product $(1 + y_2 + y_2')(1 + y_3 + y_3')$ equals zero, which implies $g(x, \tilde{y}) = g(\tilde{x}, y) = g(\tilde{x}, \tilde{y}) = g(x, y) + 1$ and thus finishes the proof.
    \qed
\end{proof}

\repeatclm{clmwindowb}
\begin{proof}
    Again, all calculations in this proof are done modulo 2 and it suffices to consider a horizontal window $w = \{(x,y), \, (x,y')\}$ due to symmetry reasons.
    Let the maps $(x, y, y') \to (x, \tilde{y}, \tilde{y}')$, $(x, y, y') \to (\tilde{x}, y, y')$ and $(x, y, y') \to (\tilde{x}, \tilde{y}, \tilde{y}')$ be denoted by $\ps{f}{10}$, $\ps{f}{01}$ and $\ps{f}{00}$, respectively. Using the identity $(1 + y_2 + y_2')(1 + y_3 + y_3') = 0$ (see the proof of Claim \ref{clmwindowa}), simple calculations yield
    $\ps{f}{10}^2(x, y, y') = (x, y, y')$,
    $\ps{f}{01}^2(x, y, y') = (x, y, y')$ and
    $\ps{f}{00}^2(x, y, y') = (x + e_1, y + e_1, y' + e_1)$,
    where $e_1 = (1, 0, 0)$. Hence, for any $a, b \in \bits$ and $y \neq y'$, the map $\ps{f}{ab}^4$ is the identity, which proves that $w \to \ps{w}{ab}$ is a bijection on the set of windows.
    \qed
\end{proof}

\repeatclm{clmweee}
\begin{proof}
   In order to prove mutual independence, we show that
    $q(x_F y_F \mid tw, \CD) = \prod_{e \in F} q(x_e y_e \mid tw, \CD)$
    holds for every subset $F \subseteq E$ of edges.
    If $\Zeros(w) \neq t_u \cup t_0$ or $x_F y_F \notin w_F$, both sides of this equality are zero. For $w, t_u, t_0, x_F y_F$ with $\Zeros(w) = t_u \cup t_0$ and $x_F y_F \in w_F$, we have
    \begin{align*}
        q(x_F y_F \mid tw, \CD)
        &= \sum_{\substack{\tilde{x} \tilde{y} \in w: \\ \tilde{x}_F \tilde{y}_F = x_F y_F}} q(\tilde{x} \tilde{y} \mid tw, \CD) \\
        &= |\{\tilde{x} \tilde{y} \in w: \tilde{x}_F \tilde{y}_F = x_F y_F\}| \cdot 2^{-|E|}
        = 2^{|E| - |F|} \cdot 2^{-|E|} = 2^{- |F|},
    \end{align*}
    where the second equality follows since $q(XY \mid tw, \CD)$ is uniform over the set $\{xy: xy \in w\}$. Applying this equation for $F = \{e\}$ yields    
    \begin{align*}
        \prod_{e \in F} q(x_e y_e \mid tw, \CD) &=
        \prod_{e \in F} q(x_{\{e \}} y_{\{e \}} \mid tw, \CD) 
        = \prod_{e \in F} 2^{-1} = 2^{- |F|},
    \end{align*}
    which proves the claim.
    \qed
\end{proof}

\repeatclm{clmblub}
\begin{proof}
    For any given $t, w, i$, let $\tilde{w}$ denote the window that fulfills $\ps{\tilde{w}^i}{00}\tilde{w}^{-i} = w$. Note that exactly one such window exists due to Claim \ref{clmwindowb}. Since $q(W \mid t)$ is uniform for any $t$, we have
     \begin{align*}
         q(xryw \mid t)
         &= q(xry) \cdot q(w \mid t) 
         = q(xry) \cdot q(W = \tilde{w} \mid t) \\
         &= q(xry) \cdot q(\ps{W^i}{00}W^{-i} = w \mid t)
         = q(xry, \ps{W^i}{00}W^{-i} = w \mid t),
     \end{align*}
    which implies $q(XRY\ps{W^i}{00}W^{-i} \mid t) = q(XRYW \mid t)$.
    Moreover, $X^iY^i$ is fully determined by $XRY$ and independent of $W$ and $\ps{W^i}{00}W^{-i}$ in $q(\cdot \mid t)$.
    Therefore, the distributions $q(X^iRY^i\ps{W^i}{00}W^{-i} \mid t)$ and $q(X^iRY^iW \mid t)$ are equal.
    Since $q(X^iRY^iW \mid t, \CD)$ is obtained from $q(X^iRY^i \ps{W^i}{00}W^{-i} \mid t, \CD_i)$ by replacing every occurrence of $\ps{W^i}{00}W^{-i}$ with $W$, it follows that these two distributions are also equal, which implies $ \Infc[q]{R}{X^iY^i}{T W, \CD} = \Infc[q]{R}{X^iY^i}{T \ps{W^i}{00}W^{-i}, \CD_i}$.
    
    Further note that for any fixed $t$, $W$ is determined by $\ps{W^i}{00}W^{-i}$ and vice versa.
    Thus, we have $\Infc[q]{R}{X^iY^i}{T \ps{W^i}{00}W^{-i}, \CD_i} = \Infc[q]{R}{X^iY^i}{TW,\CD_i}$, which concludes the proof.
    \qed
\end{proof}

\repeatclm{clmztuior}
\begin{proof}
    Since $q(zt_{-i} \mid \CD_i) > 0$, we see that there exist $w, x, y, t_u, t_0, t_i$ such that $\Zeros(x, y) = t_u \cup t_0$, $x, y \in \ps{w^i}{00}w^{-i}$, $\Zeros(w) = z$ and $t_u \cup t_0 \cup t_i \cup t_{-i} = \CT$ hold, which implies $\Zeros(\ps{w^i}{00}w^{-i}) = \Zeros(x, y) = t_u \cup t_0$.
    Since $w$ arises from $\ps{w^i}{00}w^{-i}$ by replacing $b$-windows with $\comp{b}$-windows along the path connecting the terminals in $t_i$, we conclude $z = \Zeros(w) = t_u \cup t_0 \cup t_i = \CT \setminus t_{-i}$.
    \qed
\end{proof}

\repeatclm{clmwednf}
\begin{proof}
    For any $a \in \bits$, we have
    \begin{multline*}
        p(A = a, B = 0, tw)
        = q(zt_{-i}, \CE)^{-1} \cdot q(A = a, B = 0, tw, \CE) \\
        = q(zt_{-i}, \CE)^{-1} \cdot q(XY \in \ps{w^i}{a0}w^{-i}, tw)
        = q(zt_{-i}, \CE)^{-1} \cdot q(tw) \cdot q(XY \in \ps{w^i}{a0}w^{-i}).
    \end{multline*}
    Since $q(XY \in \ps{w^i}{a0}w^{-i}) = 2^{|E|} \cdot ||S||_1^{-1} \cdot (|\Zeros(\ps{w^i}{a0}w^{-i})| - 1 + \eps$ and $\Zeros(\ps{w^i}{a0}w^{-i})$ does not depend on the choice of $a$ (by Claim \ref{clmwindowa}), this proves the claim.
    \qed
\end{proof}

\repeatclm{clmuoetzwh}
\begin{proof}
    We set $w_\CE := \ps{w^i}{00}w^{-i} \cup \ps{w^i}{01}w^{-i} \cup \ps{w^i}{10}w^{-i} \cup \ps{w^i}{11}w^{-i}$ and calculate $p(XY \mid rtw) = q(XY \mid rtw, XY \in w_\CE)$, which is a restriction of $q(XY \mid rtw)$ to a rectangle. Considering that $X$ and $Y$ are independent in $q(\cdot \mid rtw)$, it follows that $X$ and $Y$ are also independent in $p(\cdot \mid rtw)$, which implies that $X - R - Y$ holds in $p(\cdot \mid tw)$.
    Since $A$ is determined by $X$ and $B$ is determined by $Y$ in $p(\cdot \mid tw)$, this proves the statement of the claim.
    \qed
\end{proof}

\repeatclm{zwezroi}
\begin{proof}
    We have $p(A = 1, B = 0) = p(A = 0, B = 0)$ by the proof of \Cref{bigprop0}. Together with Claim \ref{clmwednf}, this yields $p(tw \mid A = 1, B = 0) = p(tw \mid A = 0, B = 0)$ for any $tw$, which proves $p(TW \mid A = a, B = 0) = p(TW \mid B = 0)$ for any $a \in \bits$.
    
    Next, for some $r, t, w, a$, let $\tilde{w}$ denote the window such that $\ps{\tilde{w}}{\bar{a}0} = \ps{w}{a0}$. Then we have
    \begin{align*}
        &\quad ~ p(rtw, A = a, B = 0) =
        q(zt_{-i}, \CE)^{-1} \cdot q(rtw, XY \in \ps{W}{a0}) \\
        &= q(zt_{-i}, \CE)^{-1} \cdot \sum_{xy \in \ps{w}{a0}} q(xrytw)
        = q(zt_{-i}, \CE)^{-1} \cdot \sum_{xy \in \ps{\tilde{w}}{\bar{a}0}} q(xryt, W = \tilde{w}) \\
        &= p(rt, W = \tilde{w}, A = \bar{a}, B = 0),
    \end{align*}
    where the third equality follows since $q(\cdot)$ is a product distribution and $q(W)$ is uniform.
    Together with $p(A = 1, B = 0) = p(A = 0, B = 0)$, this yields $p(rtw \mid A = a, B = 0) = p(rt, W = \tilde{w} \mid A = \bar{a}, B = 0)$. Since $\tilde{w}^{-i} = w^{-i}$ holds by the definition of $\tilde{w}$, we have
    $p(r \mid tw^{-i}, A = a, B = 0) = p(r \mid t\tilde{w}^{-i}, A = \bar{a}, B = 0) = p(r \mid tw^{-i}, A = \bar{a}, B = 0)$. 
    \qed
\end{proof}

\repeatclm{jjiofe}
\begin{proof}
    For a given window $w$, let $\tilde{w}$ denote the unique window which fulfills $\ps{w^i}{10}w^{-i} = \ps{\tilde{w}^i}{00}\tilde{w}^{-i}$. By Claim \ref{clmwindowb} and Claim \ref{clmwindowa}, we know that $w \mapsto \tilde{w}$ is a bijection for any $t$, and if $w$ is a $b$-window, then $\tilde{w}$ is also a $b$-window. Note that this implies $\Zeros(w) = \Zeros(\tilde{w})$.
    If $\Zeros(w) \neq z$, $t$ is not consistent with $t_{-i}$, or $xy \notin \ps{w^i}{10}w^{-i}$, we have $p(xrytw \mid A = 1, B = 0) = p(xryt, W = \tilde{w} \mid A = 0, B = 0) = 0$. Otherwise we have
    \begin{align*}
        &\quad ~ q(zt_{-i}, \CE) \cdot p(A = 1, B = 0) \cdot p(xrytw \mid A = 1, B = 0) \\
        &= q(zt_{-i}, \CE) \cdot p(xrytw, A = 1, B = 0)
        = q(xrytw, A = 1, B = 0, zt_{-i}, \CE) \\
        &= q(xrytw, xy \in \ps{w^i}{10}w^{-i})
        = q(xry) \cdot q(t) \cdot q(w),
    \end{align*}
    and with an equivalent calculation
    \[
        q(zt_{-i}, \CE) \cdot p(A = 0, B = 0) \cdot p(xryt, W = \tilde{w} \mid A = 0, B = 0) = q(xry) \cdot q(t) \cdot q(W = \tilde{w}).
    \]
    Following the proof of \Cref{bigprop0}, we see that $p(A = 1, B = 0) = p(A = 0, B = 0)$ holds. Since $q(W)$ is uniform, this shows that we have $p(xrytw \mid A = 1, B = 0) = p(xryt, W = \tilde{w} \mid A = 0, B = 0)$.
    Hence, together with \Cref{prop:infoexpected}, we get
    \begin{align*}
        &\Infc[p]{R}{X^iY^i}{TW, A = 1, B = 0} \\
        &=
        \sum_{tw} p(tw \mid A = 1, B = 0) \cdot \Infc[p]{R}{X^iY^i}{tw, A = 1, B = 0} \\
        & = \sum_{tw} p(t, W = \tilde{w} \mid A = 0, B = 0) \cdot \Infc[p]{R}{X^iY^i}{t, W = \tilde{w}, A = 0, B = 0} \\
        & = \sum_{t\tilde{w}} p(t, W = \tilde{w} \mid A = 0, B = 0) \cdot \Infc[p]{R}{X^iY^i}{t, W = \tilde{w}, A = 0, B = 0} \\
        & = \Infc[p]{R}{X^iY^i}{TW, A = 0, B = 0},
    \end{align*}
    where the third equality follows from $w \mapsto \tilde{w}$ being a bijection.
    \qed
\end{proof}

\repeatclm{hkfdsua}
\begin{proof}
    Let us denote the nodes and edges along the path between the terminals in $t_i$ as $v_1, e_1, \dots,$ $v_n, e_n, v_{n + 1}$. Recall that in $s(\cdot)$, the values of $T$, $W^{-i}$, $A$ and $B$ are fixed. This means that for every edge $e$ that does not lie on this path, the value of $g(x_e, y_e)$ is already fixed. Hence, $g(x_{e_1}, y_{e_1})$ alone determines whether $v_1$ is in $\Zeros(W)$ or not. Thus, since $Z$ is also fixed in $s(\cdot)$, it is predetermined whether $w_{e_1}$ has to be a $0$- or $1$-window. This however means that for every edge $e$ connecting to $v_2$, apart from $e_2$, the value of $g(x_e, y_e)$ is again fixed. The claim follows by repeating the same argument along the entire path.
    \qed
\end{proof}

\repeatclm{clmhkfdsud}
\begin{proof}
    First, note that $RXY$ and $W$ are independent in $q(\cdot \mid t)$, which implies that $RX^{-i}Y^{-i}$ and $W^i$ are independent in $q(\cdot \mid w^{-i}, x^iy^i, t)$. Moreover, observe that $s(RX^{-i}Y^{-i}, W^i \mid x^iy^i) = q(RX^{-i}Y^{-i}, W^i \mid x^iy^i, tw^{-i}, x^iy^i \in \ps{W^i}{a0}, X^{-i}Y^{-i} \in w^{-i})$ is the restriction of $q(RX^{-i}Y^{-i}, W^i \mid w^{-i}, x^iy^i, t)$ to the rectangle $\CS \times \CT$, where $\CS = \{rx^{-i}y^{-i}: x^{-i}y^{-i} \in w^{-i}\}$ and $\CT = \{w^i: \ps{w^i}{a0} \ni x^iy^i\}$. Thus we have $RX^{-i}Y^{-i} - X^iY^i - W^i$ in $s(\cdot)$, which implies $R - X^iY^i - W^i$.
    Recall that we have $\tilde{X}^i\tilde{Y}^i - W^i - X^iY^iR$ by the definition of $\tilde{X}^i\tilde{Y}^i$, which yields
    \begin{align*}
        &\quad ~ s(r, \tilde{x}^i\tilde{y}^iw^i \mid x^iy^i)
        = s(x^iy^i)^{-1} \cdot s(w^i) \cdot s(r, \tilde{x}^i\tilde{y}^ix^iy^i \mid w^i) \\
        &= s(x^iy^i)^{-1} \cdot s(w^i)
        \cdot s(rx^iy^i \mid w^i) \cdot s(\tilde{x}^i\tilde{y}^i \mid w^i) 
        = s(rw^i \mid x^iy^i) \cdot s(\tilde{x}^i\tilde{y}^i \mid w^i) \\
        &= s(r \mid x^iy^i) \cdot s(w^i \mid x^iy^i) \cdot s(\tilde{x}^i\tilde{y}^i \mid w^i, x^iy^i) 
        = s(r \mid x^iy^i) \cdot s(\tilde{x}^i\tilde{y}^i w^i \mid x^iy^i),
    \end{align*}
    where the fourth equation uses $R - X^iY^i - W^i$. Hence we have $R - X^iY^i - \tilde{X}^i\tilde{Y}^iW^i$ in $s(\cdot)$. Since the random variables $\tilde{X}^i\tilde{Y}^i$ and $\{X^iY^i, \tilde{X}^i\tilde{Y}^i\}$ are determined by each other in $s(\cdot \mid x^iy^i)$, this implies $R - X^iY^i - \{X^iY^i, \tilde{X}^i\tilde{Y}^i\} W^i$ in $s(\cdot)$.
    \qed
\end{proof}

\repeatclm{clmhtoewph}
\begin{proof}
    By applying the Markov chain $R - X^iY^i - \tilde{X}^i\tilde{Y}^iW^i$ (Claim \ref{clmhkfdsud}) twice, we get
    \begin{align*}
        s(r, w^i \mid x^iy^i, \tilde{x}^i\tilde{y}^i)
        &= s(\tilde{x}^i\tilde{y}^i \mid x^iy^i)^{-1} \cdot s(r, w^i, \tilde{x}^i\tilde{y}^i \mid x^iy^i) \\
        &= s(\tilde{x}^i\tilde{y}^i \mid x^iy^i)^{-1} \cdot s(r \mid x^iy^i) \cdot s(w^i, \tilde{x}^i\tilde{y}^i \mid x^iy^i) \\
        &= s(r \mid x^iy^i, \tilde{x}^i\tilde{y}^i) \cdot s(w^i \mid x^iy^i, \tilde{x}^i\tilde{y}^i)
    \end{align*}
    for any $r, w^i, x^iy^i, \tilde{x}^i\tilde{y}^i$. Moreover, we have
    \begin{align*}
        s(w^i \mid X^iY^i = x^iy^i, \tilde{X}^i\tilde{Y}^i = \tilde{x}^i\tilde{y}^i)
        &= s(w^i \mid \{x^iy^i, \tilde{x}^i\tilde{y}^i\} \subset \ps{W^i}{a0}) \\
        &= s(w^i \mid \{x^iy^i, \tilde{x}^i\tilde{y}^i\})
    \end{align*}
    for any $w^i, x^iy^i, \tilde{x}^i\tilde{y}^i$. Therefore we obtain
    \begin{align*}
        &\quad ~ s(r, w^i \mid \{x^iy^i, \tilde{x}^i\tilde{y}^i\}) \\
        &= s(X^iY^i = x^iy^i \mid \{x^iy^i, \tilde{x}^i\tilde{y}^i\}) \cdot s(r, w^i \mid X^iY^i = x^iy^i, \tilde{X}^i\tilde{Y}^i = \tilde{x}^i\tilde{y}^i) \\
        &+ s(X^iY^i = \tilde{x}^i\tilde{y}^i \mid \{x^iy^i, \tilde{x}^i\tilde{y}^i\}) \cdot s(r, w^i \mid X^iY^i = \tilde{x}^i\tilde{y}^i, \tilde{X}^i\tilde{Y}^i = x^iy^i) \\
        &= s(w^i \mid \{x^iy^i, \tilde{x}^i\tilde{y}^i\}) \\
        &\cdot [s(X^iY^i = x^iy^i \mid \{x^iy^i, \tilde{x}^i\tilde{y}^i\}) \cdot s(r \mid X^iY^i = x^iy^i, \tilde{X}^i\tilde{Y}^i = \tilde{x}^i\tilde{y}^i) \\
        &+ s(X^iY^i = \tilde{x}^i\tilde{y}^i \mid \{x^iy^i, \tilde{x}^i\tilde{y}^i\}) \cdot s(r \mid X^iY^i = \tilde{x}^i\tilde{y}^i, \tilde{X}^i\tilde{Y}^i = x^iy^i) ] \\
        &= s(w^i \mid \{x^iy^i, \tilde{x}^i\tilde{y}^i\})
        \cdot s(r \mid \{x^iy^i, \tilde{x}^i\tilde{y}^i\})
    \end{align*}
    for any $r, w^i, x^iy^i, \tilde{x}^i\tilde{y}^i$, which proves $R - \{X^iY^i, \tilde{X}^i\tilde{Y}^i\} - W^i$ in $s(\cdot)$.
    \qed
\end{proof}

\repeatclm{jlbopf}
\begin{proof}
    Let $W^{-i, 0}$ denote the entries of $W$ corresponding to edges that do not lie in the paths connecting the terminals in $t_i$ and $t_0$. Recall that $z = t_u \cup t_0 \cup t_i$ by Claim \ref{clmztuior}, which implies $|\Zeros(\ps{w^i}{ab} \ps{w^0}{cd} w^{-i, 0})| \in \{1, 3, 5\}$, where the exact value depends on how many of the two pairs $(a, b)$ and $(c, d)$ equal $(1, 1)$.
    Therefore, for any $tw$ consistent with $zt_{-i}$, we have
    \begin{align*}
        q(\CE \mid tw, \CF) 
        &= q(C = 1, D = 1 \mid tw, \CF) \\
        &= q(\CF \mid tw)^{-1} \cdot \sum_{a, b \in \bits}
        q(ab, C = 1, D = 1, \CF \mid tw) \\
        &= q(\CF \mid tw)^{-1} \cdot \sum_{a, b \in \bits}
        q(XY \in \ps{w^i}{ab} \ps{w^0}{11} w^{-i, 0} \mid tw) \\
        &= q(\CF \mid tw)^{-1} \cdot \sum_{a, b \in \bits} 2^{|E|} ||S||_1^{-1} \cdot 
        \left( |\Zeros(\ps{w^i}{ab} w^{-i})| - 1 + \eps \right) \\
        &= q(\CF \mid tw)^{-1} \cdot 2^{|E|} ||S||_1^{-1} \cdot (1 \cdot (4 + \eps) + 3 \cdot (2 + \eps)) \\
        &= q(\CF \mid tw)^{-1} \cdot 2^{|E|} ||S||_1^{-1} \cdot (10 + 4 \eps)
    \end{align*}
    and
    \begin{align*}
        &q(|\Zeros(X,Y)| = 1 \mid tw, \CF) \\
        &= q(\CF \mid tw)^{-1} \cdot \sum_{a,b,c,d \in \bits}
        q(abcd, |\Zeros(X, Y)| = 1, \CF \mid tw) \\
        &= q(\CF \mid tw)^{-1} \cdot \left[ 7 \cdot 0 + 2^{|E|}||S||_1^{-1} \cdot 9 \cdot \eps \right] 
        = q(\CF \mid tw)^{-1} \cdot 2^{|E|}||S||_1^{-1} \cdot 9\eps,     
    \end{align*}
    which implies
    \begin{align*}
        &\frac{q(|\Zeros(X,Y)| = 1 \mid zt_{-i}, \CF)}{q(\CE \mid zt_{-i}, \CF)} \\
        &= \frac{\sum_{tw} q(tw \mid zt_{-i}, \CF) \cdot q(|\Zeros(X,Y)| = 1 \mid tw, \CF)}
        {\sum_{tw} q(tw \mid zt_{-i}, \CF) \cdot q(\CE \mid tw, \CF)} \\
        &= \frac{9 \eps \cdot \sum_{tw} q(tw \mid zt_{-i}, \CF)}
        {(10 + 4 \eps) \cdot \sum_{tw} q(tw \mid zt_{-i}, \CF)}
        = \frac{9 \eps}{10 + 4 \eps} = O(\eps).
    \end{align*}
    \qed
\end{proof}

\repeatclm{thewiufd}
\begin{proof}
    As seen before in the proof of Claim \ref{clmuoetzwh}, we can observe that the set $\bigcup_{a,b,c,d \in \bits} \ps{w^i}{ab} \ps{w^0}{cd} w^{-i, 0}$ is a rectangle, which implies that $X - R - Y$ holds in $q(\cdot \mid tw, \CF)$.
    Since $AC$ is determined by $X$ and $BD$ is determined by $Y$ in $q(\cdot \mid tw, \CF)$, the statement of the claim follows.
    \qed
\end{proof}

\repeatclm{clmhkfdsub}
\begin{proof}
    Consider any $x^iy^i$ such that $s(x^iy^i) > 0$. Note that this immediately implies $x^iy^i \in \ps{W^i}{a0}$ and therefore $x^iy^i \in g^{-1}(b^i)$. Then, we have
    \begin{align*}
        &\quad ~ q(tw^{-i}a, B = 0, \CE, z) \cdot s(x^iy^i)
        = q(x^iy^i, tw^{-i} a, B = 0, \CE, z) \\
        &= \sum_{x^{-i}y^{-i}, w^i}
        q(xy, tw a, B = 0, \CE, z) \\
        &= \sum_{x^{-i}y^{-i}, w^i}
        q(xy, tw) = \sum_{x^{-i}y^{-i}, w^i}
        q(tw) \cdot ||S||_1^{-1} \cdot (2 + \eps),
    \end{align*}
    where the sums are over $x^{-i}y^{-i} \in w^{-i}$ and $w^i$ such that $x^iy^i \in \ps{w^i}{a0}$.
    Since $q(TW)$ is uniform and $|\{w^i: x^iy^i \in \ps{w^i}{a0}\}|$ is determined by $t$, the final expression does not depend on $x^iy^i$, which proves that $s(X^iY^i)$ is uniformly distributed in $g^{-1}(b^i)$.
    \qed
\end{proof}

\repeatclm{clmhkfdsuc}
\begin{proof}
    Let $w^i$ be a $\comp{b^i}$-window, which implies $\Zeros(w^iw^{-i}) = z$, and, due to Claim \ref{clmztuior}, $\Zeros(\ps{w^i}{a0} w^{-i}) = t_u \cup t_0$. Then, we have
    \begin{align*}
        q(tw^{-i}a, B = 0, \CE, z) \cdot s(w^i)
        &= q(twa, B = 0, \CE, zt_{-i}) = q(tw, XY \in \ps{w^i}{a0}w^{-i}) \\
        &= q(tw) \cdot q(XY \in \ps{w^i}{a0}w^{-i}) \\
        &= q(tw) \cdot 2^{-n} \cdot ||S||_1^{-1} \cdot (|\Zeros(\ps{w^i}{a0} w^{-i}) - 1 + \eps) \\
        &= q(t) \cdot q(w) \cdot 2^{-n} \cdot ||S||_1^{-1} \cdot (2 + \eps).
    \end{align*}
    Since $q(W)$ is uniform, this proves that $s(w^i)$ does not depend on the choice of $w^i$. Therefore, $s(W^i)$ is the uniform distribution over all $\comp{b^i}$-windows.
    \qed
\end{proof}

\end{document}